\documentclass[pra,floatfix,amsmath,superscriptaddress,twocolumn]{revtex4-1}
\usepackage{amssymb}
\usepackage{graphicx}
\usepackage{graphics}
\usepackage{amsmath}
\usepackage{amsthm}
\usepackage{color}
\usepackage{float}

\newcommand{\bea}{\begin{eqnarray}}
\newcommand{\eea}{\end{eqnarray}}
\def\bi{\begin{itemize}}
\def\ei{\end{itemize}}
\def\bc{\begin{center}}
\def\ec{\end{center}}

\def\C{\hbox{$\mit I$\kern-.7em$\mit C$}}
\def\R{\hbox{$\mit I$\kern-.6em$\mit R$}}

\def\ket#1{|#1\rangle}
\newcommand{\one}{\mbox{$1 \hspace{-1.0mm}  {\bf l}$}}

\def\ket#1{\left| #1\right>}

\usepackage{bm}
\usepackage{enumerate}
\newcommand{\braket}[2]{\left< #1 \vphantom{#2} \right| \left. #2 \vphantom{#1} \right>}
\newcommand{\braopket}[3]{\left< #1 \vphantom{#2#3} \right| #2 \left| #3 \vphantom{#1#2} \right>}

\newcommand{\hilbert}{\mathcal{H}}

\newcommand{\identity}{\one}

\newcommand{\trace}[2][]{\text{tr}_{#1}\left( #2 \right)}

\newcommand{\Span}[1]{\operatorname{span}\left\{ #1 \right\}}
\newcommand{\conv}[1]{\operatorname{conv}\left\{ #1 \right\}}

\renewcommand{\vec}[1]{\bm{#1}}

\newtheorem{theorem}{Theorem}
\newtheorem{corollary}[theorem]{Corollary}
\newtheorem{observation}[theorem]{Observation}
\newtheorem{lemma}[theorem]{Lemma}

\begin{document}
\author{M. Hebenstreit}
\affiliation{Institute for Theoretical Physics, University of
Innsbruck, Innsbruck, Austria}
\author{C. Spee}
\affiliation{Institute for Theoretical Physics, University of
Innsbruck, Innsbruck, Austria}
\author{B. Kraus}
\affiliation{Institute for Theoretical Physics, University of
Innsbruck, Innsbruck, Austria}
\title{The MES of tripartite qutrit states and pure state separable transformations which are not possible via LOCC}

\begin{abstract}
Entanglement is the resource to overcome the restriction of operations to  Local Operations assisted by Classical Communication (LOCC). The Maximally Entangled Set (MES) of states is the minimal set of n--partite pure states with the property that any truly n--partite entangled pure state can be obtained deterministically via LOCC from some state in this set. Hence, this set contains the most useful states for applications. In this work we characterize the MES for generic three qutrit states. Moreover, we analyze which generic three qutrit states are reachable (and convertible) under LOCC transformations. To this end we study reachability via separable operations (SEP), a class of operations that is strictly larger than LOCC. Interestingly, we identify a family of pure states that can be obtained deterministically via SEP but not via LOCC. To our knowledge these are the first examples of transformations among pure states that can be implemented via SEP but not via LOCC.
\end{abstract}
\maketitle

\section{Introduction}
In many quantum informational tasks such as quantum cryptography \cite{article:ekert91} and quantum secret sharing \cite{article:hillery91,article:gottesman} a scenario is considered where several parties are spatially separated and they can only communicate classically. In this scenario the parties are restricted to Local Operations and Classical Communication (LOCC). LOCC transformations and its structure have been studied extensively \cite{article:horodecki09,*[{}] [{ and references therein.}] article:everything}. The interest in these class of operations is not only due to its correspondence to the very natural physical setup mentioned above but the concept of LOCC is strongly interconnected with entanglement. In particular, entanglement is a resource when the operations are restricted to LOCC, i.e. entanglement is non-increasing under  LOCC and it can be used to overcome the restrictions set by LOCC. Entanglement plays also an important role in quantum computation (see e.g. \cite{book:jozsa98, article:briegel01}) and other fields of physics such as condensed matter physics \cite{* [{For a review on multipartite systems see for instance }] [{ and references therein.}] article:amico08}. In the bipartite case entanglement is largely understood \cite{article:horodecki09}, e.g. several entanglement measures and the states that are maximally entangled are known. This is due to the fact that in the case of bipartite pure states one can show that any LOCC transformation that is possible can be achieved by a very simple protocol \cite{article:lo01}. This allowed for a characterization of LOCC transformation between bipartite pure states known as Nielsen\textquoteright s majorization criterion \cite{article:nielsen}. \\
LOCC convertibiliy among  multipartite entangled states has only been studied for small system sizes \cite{article:turgut10,article:kintas10,article:Tajima13,article:mes,article:spee15}.
This is due to the fact that in general the structure of LOCC transformations is very subtle. It has been shown that some tasks require infinitely many rounds \cite{article:infinite}. As the mathematical characterization of LOCC transformations is so challenging a different class of operations called separable maps (SEP) that include LOCC have been studied (see \cite{article:gour} and references therein).  Although mathematically easy to characterize, SEP lacks a clear physical interpretation as it has been shown that SEP is strictly larger than LOCC \cite{article:seplocc, article:everything}, i.e., there exist separable maps that can not be implemented via LOCC. Moreover, state transformations involving mixed states have been found that can be achieved via SEP but not via LOCC \cite{article:chitambar12, article:chitambar09}. Furthermore, there exist entanglement monotones, i.e. quantities that are non-increasing on average under LOCC which have been shown to increase under SEP \cite{article:chitambar09}. Other classes of operations that have been studied are local unitaries and stochastic LOCC-operations. Local Unitaries (LUs) constitute invertible deterministic LOCC transformations and therefore interconnect states with the same amount of entanglement. Necessary and sufficient conditions for two pure $n$-qubit states to be LU-equivalent have been derived in \cite{article:krausprl10, article:krauspra10}. Under Stochastic LOCC-operations (SLOCC-operations) states that can be transformed into each other with non-zero probability  are grouped together \cite{article:slocc,article:4qubits,article:moduli}. LU- and SLOCC-classes are equivalence classes and do not allow to identify which states are more useful than the others under the restriction to LOCC. \\
 Thus, despite all the difficulties one faces when investigating LOCC transformations its characterization is necessary in order to understand and quantify multipartite entanglement. In particular, the knowledge about which deterministic LOCC transformations are possible allows to compute operational entanglement measures \cite{article:schwaiger15, article:sauerwein15} and identify the most useful states for applications. In the bipartite case the maximally entangled state, i.e. the state which can not be obtained via LOCC (excluding LUs) from any other state but any other state can be reached via LOCC from this state, and several applications based on it such as teleportation \cite{article:teleportation} and cryptography \cite{article:ekert91} are well known. In the multipartite setting it is no longer a single state but a set of states, the Maximally Entangled Set (MES) \cite{article:mes} that is an optimal resource under LOCC. In particular, the MES of $n$ parties is given by the minimal set of pure $n$-partite states such that any other truly $n$-partite entangled pure state can be reached via LOCC from a state in the MES. Hence, the MES contains the most useful states for applications. In particular, any protocol can be based on a state in the MES (as any state can be reached via deterministic LOCC from a state in the MES) and thus identifying the MES could guide the path to new applications. \\
In \cite{article:mes,  article:spee15} we have identified the MES for three and four qubit states. Whereas for three qubits the MES is of measure zero and any state can be converted via LOCC deterministically into some other state, one observes a completely different picture for four qubits. In this case the MES is of full measure. This is due to the fact that almost all states are isolated, i.e. they can neither be reached nor converted via deterministic LOCC transformations (excluding LUs). The rare instances of non-isolated states that are in the MES are of particular interest as these are the most useful states for entanglement manipulation. They constitute a zero measure subset and exhibit a particularly simple parametrization. In the three and generic four qubit case we found that all transformations between pure truly multipartite entangled states that can be done via SEP can also be implemented via LOCC \cite{article:mes,article:schwaiger15, article:sauerwein15}. As the picture changes drastically when going from the three qubit to the four qubit case the question arises whether an increase in the local dimension would reveal a different behavior and new insight into the possible structures of LOCC transformations. This is precisely the motivation of this work, where we identify the transformability properties of generic three qutrit states under SEP and under LOCC. In particular, we identify the MES. \\
Interestingly, we find that, although the three qutrit case resembles very much the generic four qubit case, it brings a new feature into the play. In particular, we observe the first examples of pure state transformations that can be implemented via SEP but not via LOCC. Moreover, these states are reachable via SEP but not via LOCC (even if one allows for infinitely many rounds in the protocol).\\

The outline of this paper is the following: First we present the notation, the precise definition of the MES and the methods that we use in order to characterize SEP and LOCC convertibility. Then we show which generic three qutrit states can be reached via SEP and we present the first example of transformations between pure states that can be accomplished via SEP but not via LOCC. Finally, we characterize the MES for generic three qutrit states and show which states are convertible via LOCC.

\section{Notation and known results}
Let us first generalize the term truly $n$-partite entanglement to higher-dimensional systems. For a $n$-partite qubit state, the state is either truly $n$-partite entangled (all reduced density matrices are of full rank) or the state is biseparable. For $d$-dimensional systems with $d>2$ the grading is finer as the ranks of the reduced density matrices can range from $1$ to $d$.
We call a state $\ket{\phi} \in \hilbert$, where $\hilbert = \mathbb{C}^{d_1} \otimes \mathbb{C}^{d_2} \otimes \ldots \otimes \mathbb{C}^{d_n}$ truly $n$-partite, $\vec{c}$-level entangled, where $\vec{c} = (c_1, c_2, \ldots, c_n)$ with $c_i \leq d_i$, if the reduced density matrix of system $i$ has rank $c_i$ for all $i \in \left\{1, \ldots, n\right\}$.

Equivalence classes containing states which are equivalent under Stochastic LOCC (SLOCC) are called SLOCC classes. Two states $\ket{\phi}$ and $\ket{\psi}$ are equivalent under SLOCC if they can be transformed into each other with non vanishing probability of success using local operations, i.e., if there exists a LOCC protocol in which at least in one branch $\ket{\phi}$ is transformed into $\ket{\psi}$ and vice versa. Here and in the following one branch is determined by a specific sequence of measurement outcomes.
Hence, two states $\ket{\phi}$ and $\ket{\psi}$ are equivalent under SLOCC iff they are related by local invertible operators, i.e., $\ket{\phi} \propto A \otimes B \otimes C \ket{\psi}$, where $A$, $B$, and $C$ are invertible.
We will investigate deterministic LOCC transformations between generic three qutrit states that are truly $3$-partite, $\vec{d}$-level entangled for $\vec{d} = (3,3,3)$. In this case, LOCC transformations are only possible within SLOCC classes. We denote by $\ket{\psi}$ a particular representative (see below) of a SLOCC class, which we call a seed state in the following. Moreover, we will denote the initial state by $g \ket{\psi}$ and the final state by $h \ket{\psi}$, where $g$ and $h$ are local, invertible operators, i.e., $g = g_1 \otimes g_2 \otimes g_3$ and $h = h_1 \otimes h_2 \otimes h_3$ with $g_i, h_i \in \operatorname{GL}(3, \mathbb{C})$. Furthermore, let us introduce the notation $G = G_1 \otimes G_2 \otimes G_3 = g^\dagger g$ and $H = H_1 \otimes H_2 \otimes H_3 = h^\dagger h$.

Let us now review the definition of the MES and the methods used to determine the MES for 3 qubit and generic 4 qubit states, which are also applicable to the three qutrit case.

In the bipartite $d$-level setting the maximally entangled state $\ket{\phi^+_d} = \frac{1}{\sqrt{d}} \sum_{i=0}^{d-1} \ket{i i}$ plays a very special role as this state, on the one hand, cannot be reached from any LU-inequivalent state via deterministic LOCC and, on the other hand, can be converted into any other state via deterministic LOCC. This is due to the fact that the vector containing the Schmidt coefficients of $\ket{\phi^+_d}$ is majorized by any other vector containing sorted Schmidt coefficients, and therefore from the majorization criterion \cite{article:nielsen} it follows that $\ket{\phi^+_d}$ can indeed be converted into any other state. Moreover, the majorization criterion implies that this state cannot be reached from any other LU-inequivalent state.
In contrast to the bipartite case, there is no such maximally entangled state which can be transformed to any other state via LOCC in the multipartite setting. However, the idea of the MES is precisely to generalize the notion of \emph{maximal entanglement} to the multipartite scenario \cite{article:mes}. There, the term maximal entanglement is not associated to a single state, but to a set of states, the \emph{maximally entangled set}.	
	It is defined as follows.
	The $MES_{\vec{d}}$ is a set of truly $n$-partite, $\vec{d}$-level entangled states with the following two properties.
		\begin{enumerate}[(i)]
			\item No state in $MES_{\vec{d}}$ can be reached from any other truly $n$-partite, $\vec{d}$-level entangled, LU-inequivalent state via LOCC and
			\item for any pure, truly $n$-partite, $\vec{d}$-level entangled entangled state $\ket{\phi}$, there exists a state $\ket{\psi}$ in the $MES_{\vec{d}}$ such that $\ket{\psi}$ can be transformed into $\ket{\phi}$ via LOCC.
		\end{enumerate}
		The $MES_{\vec{d}}$ is hence the minimal \emph{set} of truly $n$-partite, $\vec{d}$-level entangled states which is maximally useful for quantum information tasks.
	
	To characterize the MES, one has to know which LOCC transformations are possible. In contrast to the bipartite case, where all possible LOCC transformations have been characterized by Nielsen \cite{article:nielsen}, such a characterization is much more demanding in the multipartite setting. This is due to the complex structure of LOCC, including possibly infinitely many rounds of communication \cite{article:infinite}.
	Despite these difficulties it is essential to know which LOCC transformations are possible to characterize entanglement as it is the resource to overcome LOCC.
In order to characterize LOCC transformations we will consider convertibility under Separable Operations (SEP), a class of quantum operations which is strictly larger than LOCC \cite{article:seplocc}.
		A completely positive, trace preserving map $\Lambda$ that can be written as
		\begin{align}
			\Lambda(\rho) = \sum_k M_k \rho M_k^\dagger \ \text{with} \  M_k = A_1^{(k)} \otimes   \cdots \otimes A_n^{(k)},
		\end{align}
		is called separable. All such $\Lambda$ form the set of Separable Operations (SEP).

In contrast to LOCC, separable operations are mathematically much more tractable, and as SEP contains LOCC,  SEP-convertibility of a state $\ket{\psi}$ to a state $\ket{\phi}$ is a necessary condition for LOCC-convertibility. However, it has been shown that LOCC is strictly contained in SEP \cite{article:seplocc}, i.e., there exist SEP-maps which cannot be implemented via LOCC. Thus, the existence of a SEP transformation is a necessary but not a sufficient condition for the existence of a corresponding LOCC transformation. Indeed, in section \ref{sec:sepneqlocc}, we will give examples for pure state transformations which are possible via SEP but not via LOCC. Note that we cannot leave the SLOCC class with LOCC transformations without going to a state which is not truly $n$-partite, $\vec{d}$-level entangled any more.

Let us now review the results on SEP obtained in \cite{article:gour}. In order to state the necessary and sufficient conditions for the existence of SEP transformations within SLOCC classes, which has been derived there, we denote by $S(\ket{\psi}) = \left\{S: S \ket{\psi} = \ket{\psi}, \;S = S_1 \otimes \ldots \otimes S_n, \;S_i \in \operatorname{GL}(d_n,\mathbb{C})\right\}$ the set of local symmetries of the state $\ket{\psi}$.
\begin{theorem}[\cite{article:gour}]
	\label{theo:gour}
			Let $\ket{\psi_1} = g \ket{\psi} \in \hilbert$ and $\ket{\psi_2} = h \ket{\psi} \in \hilbert$. Then
			$\ket{\psi_1}$ can be converted into $\ket{\psi_2}$ via SEP iff there exists a finite index set $I$, probabilities $\left\{ p_k\right\}_{k \in I}$ with $p_k \geq 0$ and $\sum_{k\in I} p_k = 1$  and symmetries of the seed state $\mathcal{S}_k \in S(\ket{\psi})$ such that
			\begin{align}
				\label{equ:gourgen}		
				\sum_{k \in I} p_k \mathcal{S}_k^\dagger H \mathcal{S}_k = r G
			\end{align}
			with $H = H_1 \otimes \ldots \otimes H_n =  h^\dagger h$, $G =  G_1 \otimes \ldots \otimes G_n  = g^\dagger g$, and $r = \left|\left| \, \ket{\psi_2 }\,\right|\right|^2 / \left| \left| \,\ket{\psi_1 }\,\right|\right|^2$.
		\end{theorem}
		Note the important role the symmetries of the seed state play in SEP transformations.		
In order to transform $\ket{\psi_1}$ into $\ket{\psi_2}$ via SEP, the parties have to implement a POVM whose elements depend on the matrices $G_i$ and $H_i$, as the explicit SEP map which transforms the state $\ket{\psi_1}$ to $\ket{\psi_2}$ is given by
\begin{align}
	\label{equ:sepmap}
	\Lambda_{SEP}\left( \rho \right) = \sum_{k \in I} M_k \rho M_k^\dagger,\ \text{where} \  M_k = \frac{1}{\sqrt{r}} \sqrt{p_k} h \mathcal{S}_k g^{-1}.
\end{align}
We will see later on that if such a transformation is possible via LOCC, then in order to transform $\ket{\psi_1}$ into $\ket{\psi_2}$ the parties can locally implement the POVM corresponding to the given SEP map.

In order to determine the MES for generic three qutrit states, we will follow the methods introduced in \cite{article:mes} for the four qubit case. Using the characterization of generic three qutrit SLOCC classes in \cite{article:moduli} and the representatives (seed states) given therein, we will first determine the symmetries of the seed states. We will find that there are, as in the four qubit case, only finitely many and that all of them are unitary. Using these facts we will then derive simple necessary conditions for SEP-convertibility by calculating partial traces over Eq. (\ref{equ:gourgen}). We will introduce a standard form for generic three qutrit states. Subsequently, we will use the derived conditions as well as the standard form to characterize states which are reachable via SEP. Using this characterization, we will show which of the states that are reachable via SEP are indeed reachable via LOCC by constructing the corresponding LOCC protocols which are of a simple form. However, we will show, that, surprisingly, a family of states which can be reached via SEP cannot be reached via LOCC, even after infinitely many rounds. We will use these results on LOCC to derive the MES for generic tripartite qutrit states. Finally, we will characterize states which are convertible under LOCC and identify isolated states, i.e., states which can neither be reached from, nor converted to another LU-inequivalent state.

\section{Seed states and their symmetries}
\label{sec:seed}
	In this section, we will first present the seed states of the generic SLOCC classes derived in \cite{article:moduli} and determine their local symmetries. Then we will derive the necessary and sufficient conditions for SEP-convertibility.

	Recall that states $\ket{\phi}$ and $\ket{\psi}$ are in the same SLOCC class iff they are related by local invertible operators, i.e., $\ket{\phi} \propto A \otimes B \otimes C \ket{\psi}$, where $A, B, C \in \text{SL}(3, \mathbb{C})$. In \cite{article:moduli}, it has been shown that generic states belong to SLOCC classes with the seed states
\begin{align}
	\label{equ:seed}
			\ket{\psi} &=  a \left(\ket{000} + \ket{111} + \ket{222}\right) \nonumber\\
				      &\quad +	 b \left(\ket{012} + \ket{201} + \ket{120}\right)  \nonumber\\
				      &\quad +	 c \left(\ket{021} + \ket{210} + \ket{102}\right),
		\end{align}
 		where $a, b, c \in \mathbb{C}$. Therefore, the generic SLOCC classes of three qutrit states are characterized by 4 real parameters.

In the following lemma, we give the symmetries, $S(\ket{\psi})$, of the seed states $\ket{\psi}$ of the generic SLOCC classes. Note that several choices of the parameters $a$, $b$, and $c$ have to be excluded as there possibly exist additional symmetries in these cases. We exclude the cases where $a=0$, $b=0$, $c=0$, $a^3 + b^3 + c^3 = 0$, $(a^3 + b^3 + c^3)^3 = (3 a b c)^3$, $a^9=b^9$, $a^9=c^9$, $b^9=c^9$, $a+b+c=0$, $a+ \omega b+c=0$, $a+ \omega^2 b+c=0$, $a+b+ \omega c=0$, $a+b+ \omega^2 c=0$, $a+\omega b+ \omega^2 c=0$, $a+\omega^2 b+ \omega c=0$, $a b + b c  + c a = 0$, $a b + \omega b c  + c a = 0$, $a b + \omega^2 b c  + c a = 0$, $a b + b c  + \omega c a = 0$, $a b + b c  + \omega^2 c a = 0$, $a b + \omega b c  + \omega^2 c a = 0$, or $a b + \omega^2 b c  + \omega c a = 0$ 
where $\omega = e^{i\frac{2\pi}{3}}$. In the following, we will call a 3 qutrit seed state (3 qutrit state) generic if it is (belongs to a SLOCC class of some seed state) given in Eq. (\ref{equ:seed}) where none of the above conditions is fulfilled, respectively.

\begin{lemma}
\label{lemma:symmetries}
Let $\ket{\psi}$ be a generic three qutrit seed state. The symmetries of $\ket{\psi}$ are given by the so-called generalized Pauli matrices,
		\begin{align}		
				S(\ket{\psi})= &\left\{\identity_3^{\otimes 3}, (X)^{\otimes 3}, (X^2)^{\otimes 3}, (Z)^{\otimes 3}, (Z^2)^{\otimes 3}, (XZ)^{\otimes 3}, \right.\nonumber\\
						 &\quad  \left.  (XZ^2)^{\otimes 3}, (X^2Z)^{\otimes 3}, (X^2Z^2)^{\otimes 3}   \right\},
		\end{align}
		where
		\begin{align}
				X = \begin{pmatrix}
					 0 & 1 & 0 \\
					 0 & 0 & 1 \\
					 1 & 0 & 0
					 \end{pmatrix} 	
				\quad\text{and}\quad
				Z = \begin{pmatrix}
					 1 & 0 & 0 \\
					 0 & \exp{i\frac{2\pi}{3}} & 0 \\
					 0 & 0 & \exp{i\frac{4\pi}{3}}
					 \end{pmatrix}. 	
		\end{align}
\end{lemma}
This lemma is proven in the appendix. Note that $S(\ket{\Psi})\backslash \left\{\identity_3^{\otimes 3} \right\}$ corresponds to the set of all the generators of SU(3).  Note further that there are only finitely many symmetries, and note that they are unitary (as it was also the case for the generic four qubit SLOCC classes).
	In the following, we will denote the generalized Pauli matrices by $S_{\vec{k}} = X^{k_1} Z^{k_2}$, where $\left( \begin{smallmatrix}  k_1 \\  k_2  \end{smallmatrix}\right) \equiv \vec{k} \in \mathbb{Z}_3^2$. The symmetries are then $\left\{S_{\vec{k}}^{\otimes 3}\right\}_{\vec{k} \in \mathbb{Z}_3^2}$. We will use the standard calculation rules for the vector space $\mathbb{Z}_3^2$, in particular, we denote the inverse element of $\vec{k}$ with respect to addition by $-\vec{k}$.

	  Note that for all $\vec{l}$, $\vec{m}$ we have that
	 \begin{align}
	 	\label{equ:pauligroup}
	 	S_{\vec{l}} S_{\vec{m}} \propto S_{\vec{l}+\vec{m}}
	 \end{align}
	 as $S_{\vec{l}} S_{\vec{m}} = X^{l_1} Z^{l_2}  X^{m_1} Z^{m_2} \propto X^{l_1 + m_1} Z^{l_2 + m_2} =  S_{\vec{l}+\vec{m}}$, where we used that $Z^{l_2} X^{m_1} \propto X^{m_1} Z^{l_2}$. In particular, it holds that
	 \begin{align}
	\label{equ:conjugation}
	 S_{\vec{l}}^\dagger S_{\vec{k}} S_{\vec{l}} = e^{i \phi_{\vec{k} \vec{l}}} S_{\vec{k}},
\end{align}
with $\phi_{\vec{k} \vec{l}} \in \{1, e^{i\frac{2 \pi}{3}}, e^{i\frac{4 \pi}{3}}\}$   for any $\vec{l}$, $\vec{k}$.	
	Note further that considering the scalar product $(A,B) = \trace{A^\dagger B}$, the generalized Pauli matrices  $\left\{ S_{\vec{k}} \right\}_{\vec{k} \in \mathbb{Z}_3^2}$ form an orthogonal basis for the $3\times 3$ matrices.

Due to Eq. (\ref{equ:gourgen}) the operators of interest are the strictly positive operators of the form $G_i = g_i^\dagger g_i$. Thus, without loss of generality we normalize the operators $G_i$ such that $\trace{G_i} = 1$. Obviously, $H_i$ can be normalized in the same way.
Using that the generalized Pauli matrices form an orthogonal basis and that $S_{\vec{0}} = \identity$, where $ \vec{0} = \left( \begin{smallmatrix}  0 \\  0  \end{smallmatrix}\right)$, we can rewrite $G_i$ as
	\begin{align}		
		\label{equ:pauli}
		G_i = \sum_{\vec{k} \in \mathbb{Z}_3^2 } g_{\vec{k}}^{(i)} S_{\vec{k}} = \frac{1}{3} \identity + \sum_{\vec{k} \in \mathbb{Z}_3^2 \backslash \{\vec{0}\}} g_{\vec{k}}^{(i)} S_{\vec{k}}.
	\end{align}
For later use we define the vector of the coordinates of $G_i$ with respect to the basis elements $\left\{ S_{\vec{k}} \right\}_{\vec{k} \neq \vec{0}}$  as
\begin{align}
\label{equ:gvec}
	\vec{g}^{(i)} = \left(  g_{\left( \begin{smallmatrix}  1 \\  0  \end{smallmatrix}\right)}^{(i)}, g_{\left( \begin{smallmatrix}  2 \\  0  \end{smallmatrix}\right)}^{(i)}, g_{\left( \begin{smallmatrix}  0 \\  1  \end{smallmatrix}\right)}^{(i)}, g_{\left( \begin{smallmatrix}  0 \\  2  \end{smallmatrix}\right)}^{(i)}, g_{\left( \begin{smallmatrix}  1 \\  1  	\end{smallmatrix}\right)}^{(i)}, g_{\left( \begin{smallmatrix}  2 \\ 2  \end{smallmatrix}\right)}^{(i)}, g_{\left( \begin{smallmatrix}  2 \\ 1  \end{smallmatrix}\right)}^{(i)} , g_{\left( \begin{smallmatrix}  1 \\  2  \end{smallmatrix}\right)}^{(i)}   \right)^T,
\end{align}
which is contained in $\mathbb{C}^8$.
Using the fact that $G_i$ must be hermitian and $S_{\vec{k}}^\dagger \propto S_{-\vec{k}}$ leads to the observation that the vectors $\vec{g}^{(i)}$ cannot be arbitrary. In fact, defining phases $\nu_{\vec{k}}$ such that
	\begin{align}		
		S_{\vec{k}}^\dagger = e^{i \nu_{\vec{k}}} S_{-\vec{k}},
	\end{align}
we get
	\begin{align}		
		G_i =   \frac{1}{3} \identity + \sum_{\vec{k}} {g_{\vec{k}}^{(i)}}^* S_{\vec{k}}^\dagger
			= \frac{1}{3} \identity + \sum_{\vec{k}}  {g_{-\vec{k}}^{(i)}}^* e^{i \nu_{-\vec{k}}} S_{\vec{k}}.
	\end{align}
Hence, comparing Eq. (\theequation) with Eq. (\ref{equ:pauli}) we have
	\begin{align}		
		\label{equ:nu}
		g_{\vec{k}}^{(i)} = {g_{-\vec{k}}^{(i)}}^*  e^{i \nu_{-\vec{k}}}.
	\end{align}
The vector $\vec{h}^{(i)}$ is defined in an analogous way, and obviously, the same restrictions have to hold.

Let us remark here that if the operator $G_i$ is of the special form  $G_i \in \Span{\identity, S_w, S_{-w}}$ with $\vec{w} \in \mathbb{Z}_3^2$, $\vec{w} \neq \vec{0}$, then also $g_i$ with  $g_i ^\dagger g_i = G_i$ can be chosen such that  $g_i \in \Span{\identity, S_{\vec{w}}, S_{-\vec{w}}}$ as the following observation shows.

\begin{observation}
\label{observation:span}
Let $M \in \text{GL}(3, \mathbb{C})$ be such that $M>0$ and $M^\dagger = M$ and $\vec{w} \in \mathbb{Z}_3^2 \setminus \{\vec{0}\}$ such that $M \in \Span{\identity, S_{\vec{w}}, S_{-\vec{w}}}$. Then there exists $m \in  \text{GL}(3, \mathbb{C})$ such that $m^\dagger m = M$ and  $m \in \Span{\identity, S_{\vec{w}}, S_{-\vec{w}}}$.
\end{observation}
\begin{proof}
As $\identity$, $S_{\vec{w}}$, and $S_{-\vec{w}}$ can be simultaneously diagonalized using some unitary $U$ we have that $M$ can be written as $M = U^\dagger D U$, where $D$ is diagonal and all the entries of $D$ are positive, as $M > 0$. In particular, using Eq. (\ref{equ:nu}) we obtain
\begin{align}
	M &= a_1 \identity + a_2 S_{\vec{w}} + a_2^{*} e^{i \nu_{\vec{w}}} S_{\vec{-w}}  \nonumber \\
	&= U^{\dagger} \left(  a_1 \identity + a_2 D_1 + a_2^{*} D_1^{*} \right) U,
\end{align}
where $D_1$ is diagonal. Let us now define $m = U^\dagger \sqrt{D} U$, where $\sqrt{D}$ is a diagonal matrix with the square roots of the elements of $D=a_1 \identity + a_2 D_1 + a_2^{*} D_1^{*}$ in its diagonal. As $\left\{    \identity, D_1, D_1^{*}    \right\}$ are linearly independent, any diagonal matrix, in particular $\sqrt{D}$, can be written as a linear combination of these three matrices. Hence, $m$ can be written as a linear combination of $\identity$, $S_{\vec{w}}$ and $S_{-\vec{w}}$ and $m^\dagger m = M$, which completes the proof.
\end{proof}
We will make use of this observation later on. Moreover, we will denote $g_i$ ($h_i$) which are restricted to $\Span{\identity, S_{\vec{w}}, S_{-\vec{w}}}$ by $g_{\vec{w}}^i $ ($h_{\vec{w}}^i $) respectively.

Let us now introduce a standard form for generic three qutrit states $g \ket{\psi}$. In order to get rid of the local unitary freedom in $g$, we can choose $g$ such that $g>0$ and $g$ is diagonal in the same basis as $G = g^\dagger g$. However, $G$ is not unique, as conjugation of $G$ with a symmetry of the seed state does not change the state we are referring to. This is due to the fact that for all $\vec{l}$ we have $g \ket{\psi} = g \; S_{\vec{l}}^{\otimes 3} \ket{\psi}$ and therefore $G$ and $ \left( S_{\vec{l}}^\dagger \right)^{\otimes 3} G \left(S_{\vec{l}}\right)^{\otimes 3} $ refer to the same state. Conjugation of $G$ with $\left(S_{\vec{l}}\right)^{\otimes 3}$ changes the phases of the entries in the vectors, $\arg{g_{\vec{k}}^{(i)}}$,  by $\phi_{\vec{k} \vec{l}}$ for all $i$. Note that for fixed $\vec{l} \neq \vec{0}$, always two of the entries are unchanged, three are multiplied by a phase $e^{i\frac{2 \pi}{3}}$ and the remaining three are multiplied by a phase $e^{- i \frac{2 \pi}{3}}$ \footnote{Note that the phase $g_{\vec{k}}^{(i)}$ is multiplied with is always the complex conjugation of the phase $g_{-\vec{k}}^{(i)}$ is multiplied with.}. We can use this freedom to bring $G$ into a unique form by restricting the range of the phases for certain entries in the coordinate vector for one particle. We choose here $\arg{g_{\left( \begin{smallmatrix}  1 \\  0  \end{smallmatrix}\right)}^{(1)}} \in \left[ 0, \frac{2 \pi}{3}\right)$ and $\arg{g_{ \left( \begin{smallmatrix}  0 \\  1  \end{smallmatrix} \right)}^{(1)}} \in \left[ 0, \frac{2 \pi}{3}\right)$. Note that in case one or both of those entries vanish one has to go through the other entries (parties) in a certain order and fix the range of the phases for the first non vanishing entries (parties). Furthermore, the seed states can be chosen uniquely. As can be easily seen, this leads to a unique standard form. As any generic state can be transformed into its unique standard form by LU we have the following Lemma.
\begin{lemma}
	Two generic states in $\mathbb{C}^3 \otimes \mathbb{C}^3 \otimes \mathbb{C}^3$ are LU-equivalent iff their standard forms as defined above coincide.
\end{lemma}
Counting parameters, one obtains that 4 real parameters are required for the seed state and 4 complex parameters for each $G_i$ in order to identify a generic state (up to LUs). This amounts in total to 28 real parameters which suffice to parametrize LU--equivalence classes of tripartite qutrit states \footnote{It can be easily verified by a simple counting argument that indeed 28 parameters are needed to parametrize the LU--equivalence classes of tripartite qutrit states, which coincides with this result as we are considering the generic SLOCC classes.}.

Let us now derive simple necessary and sufficient conditions for the existence of SEP transformations among two states. Using the symmetries of the seed state in Theorem \ref{theo:gour} we obtain
			\begin{align}		
				\label{equ:gour}
				\sum_{\vec{k} \in \mathbb{Z}_3^2} p_{\vec{k}} \left(S_{\vec{k}}^\dagger\right)^{\otimes 3} H \left( S_{\vec{k}} \right)^{\otimes 3} = r G.
			\end{align}
			with  $H = H_1 \otimes H_2 \otimes H_3$ and $G =  G_1 \otimes G_2 \otimes G_3$ and some probability distribution $\{p_{\vec{k}}\}$
as necessary and sufficient conditions for the existence of SEP transformations mapping $g\ket{\psi}$ into $h\ket{\psi}$. By taking the trace over the left and the right hand side of  Eq. (\ref{equ:gour}) we get $r=1$. We also take partial traces of the left and the right hand side of  Eq.  (\ref{equ:gour}) over all parties but $i$ to get
	\begin{align}		
		\label{equ:gi}
		\sum_{\vec{k} \in \mathbb{Z}_3^2} p_{\vec{k}} S_{\vec{k}}^\dagger H_i S_{\vec{k}} = G_i.
	\end{align}
Inserting Eq. (\theequation) in the right hand side of Eq.  (\ref{equ:gour}) we obtain that a state $h \ket{\psi}$ is reachable via SEP iff there exists a probability distribution $\{p_{\vec{k}}\}$ such that 
	\begin{align}		
		&\sum_{\vec{k}} p_{\vec{k}} \left(S_{\vec{k}}^\dagger\right)^{\otimes 3}H \left(S_{\vec{k}} \right)^{\otimes 3}
		= \sum_{\vec{k}_1,\vec{k}_2,\vec{k}_3} p_{\vec{k}_1} p_{\vec{k}_2} p_{\vec{k}_3} \nonumber\\ &\quad \times  \left( S_{\vec{k}_1}^\dagger \otimes S_{\vec{k}_2}^\dagger \otimes S_{\vec{k}_3}^\dagger \right) H \left( S_{\vec{k}_1} \otimes  S_{\vec{k}_2} \otimes  S_{\vec{k}_3} \right) .
	\end{align}
Note that in these equations we substituted $G$, such that the equations only depend on $H$, i.e., they only depend on the final state after applying the SEP transformation.

To simplify notation we define $h_{\vec{0}}^{(i)} = \frac{1}{3}$ for all $i$. Note that the vector $\vec{h}^{(i)}$ is still eight-dimensinal as defined in Eq. (\ref{equ:gvec}) for $\vec{g}^{(i)}$.
Let us now write $H_i$ in Eq. (\theequation) in the basis of the generalized Pauli matrices as in Eq. (\ref{equ:pauli}) to obtain
	\begin{align}		
		\label{equ:sepcon1}
		&\sum_{\vec{l} \vec{m} \vec{n}} \sum_{\vec{k}} p_{\vec{k}} \;   h_{\vec{l}}^{(1)} h_{\vec{m}}^{(2)} h_{\vec{n}}^{(3)} \left(S_{\vec{k}}^\dagger\right)^{\otimes 3} \left( S_{\vec{l}} \otimes S_{\vec{m}} \otimes S_{\vec{n}}  \right) S_{\vec{k}}^{\otimes 3} \nonumber\\
		& \ = \sum_{\vec{l} \vec{m} \vec{n}} \sum_{\vec{k}_1,\vec{k}_2,\vec{k}_3} p_{\vec{k}_1} p_{\vec{k}_2} p_{\vec{k}_3}   \;  h_{\vec{l}}^{(1)} h_{\vec{m}}^{(2)} h_{\vec{n}}^{(3)} \left( S_{\vec{k}_1} \otimes S_{\vec{k}_2} \otimes S_{\vec{k}_3} \right)^\dagger \nonumber\\
		& \qquad  \times  \left( S_{\vec{l}} \otimes S_{\vec{m}} \otimes S_{\vec{n}} \right) \left(  S_{\vec{k}_1} \otimes  S_{\vec{k}_2} \otimes  S_{\vec{k}_3} \right).
	\end{align}

Using Eq. (\ref{equ:conjugation}) in Eq. (\theequation) and the fact that the set $\{S_{\vec{k}}\}_{\vec{k}}$ forms an orthogonal basis and therefore the condition given in Eq. (\ref{equ:sepcon1}) has to hold for each of the coordinates with respect to that basis, we obtain the following necessary and sufficient condition
	\begin{align}		
		\label{equ:sepconfinal1}
		  \eta_{\vec{l}} \eta_{\vec{m}} \eta_{\vec{n}} &= \eta_{\vec{l}+\vec{m}+\vec{n}} \ \forall {\vec{l} \vec{m} \vec{n} \ \text{s.t.} \ h_{\vec{l}}^{(1)} h_{\vec{m}}^{(2)} h_{\vec{n}}^{(3)} \neq 0},
	\end{align}
where $\eta_{\vec{l}} \equiv \sum_{\vec{k}} p_{\vec{k}} e^{i \phi_{\vec{l} \vec{k}}}$. We used here that $e^{i (\phi_{\vec{lk}}+ \phi_{\vec{mk}})}=e^{i \phi_{(\vec{l+m})\vec{k}}}$, which follows from Eq. (\ref{equ:pauligroup}) and Eq. (\ref{equ:conjugation}) \footnote{More precisely, we have due to Eq. (\ref{equ:pauligroup}) that for all $\vec{l}$, $\vec{m}$ and $\vec{k}$ $S_{\vec{l}} S_{\vec{m}} \propto S_{\vec{l}+\vec{m}}$ and hence Eq. (\ref{equ:conjugation}) implies on the one hand, $S_{\vec{k}}^\dagger S_{\vec{l}} S_{\vec{m}} S_{\vec{k}} = S_{\vec{l}} S_{\vec{m}} e^{i  \phi_{(\vec{l} + \vec{m}) \vec{k}}}$ as $S_{\vec{k}}^\dagger S_{\vec{l} + \vec{m}}  S_{\vec{k}} = S_{\vec{l} + \vec{m}}  e^{i  \phi_{(\vec{l} + \vec{m}) \vec{k}}}$, but on the other hand, $S_{\vec{k}}^\dagger S_{\vec{l}} S_{\vec{m}} S_{\vec{k}} = S_{\vec{k}}^\dagger S_{\vec{l}} S_{\vec{k}} S_{\vec{k}}^\dagger S_{\vec{m}} S_{\vec{k}}=  S_{\vec{l}} S_{\vec{m}} e^{i \left[ \phi_{\vec{l} \vec{k}} + \phi_{\vec{m} \vec{k}} \right]}$. Thus, we have $S_{\vec{l}} S_{\vec{m}} e^{i  \phi_{(\vec{l} + \vec{m}) \vec{k}}  } = S_{\vec{l}} S_{\vec{m}} e^{i \left[ \phi_{\vec{l} \vec{k}} + \phi_{\vec{m} \vec{k}} \right]}$.}
Note that
\begin{align}
	\label{equ:eta0}
	\eta_{\vec{0}} = \sum_{\vec{i}} p_{\vec{i}} = 1 \quad \text{and} \quad \eta_{\vec{k}}^* = \eta_{-\vec{k}},
\end{align}
where $*$ denotes the complex conjugate. Hence, there are only four independent $\eta_{\vec{k}}$. Here and in the following, we will call $\eta_{\vec{k}}$ and $\eta_{\vec{l}}$ independent if $\vec{k}, \vec{l} \neq \vec{0}$ and $\vec{k} \neq \pm \vec{l}$.
Note further that as $\eta_{\vec{k}} \in \conv{1, e^{i\frac{2 \pi}{3}}, e^{i\frac{4 \pi}{3}}} \ \forall {\vec{k}}$, where by $\operatorname{conv}$ we denote the convex hull, we have  \begin{align}
	\label{equ:etaabs}
	  |\eta_{\vec{k}}| \leq 1 \ \forall {\vec{k}}.
\end{align}

Eq. (\ref{equ:sepconfinal1}) becomes particularly simple if we choose one of the indices $\vec{l}$, $\vec{m}$, or $\vec{n}$ equal to $\vec{0}$. We obtain
	\begin{align}		
		\label{equ:sepconfinal2}
		\eta_{\vec{l}} \eta_{\vec{m}} &= \eta_{\vec{l}+\vec{m}} \ \forall {\vec{l} \vec{m} \ \text{s.t.}\ h_{\vec{l}}^{(i)} h_{\vec{m}}^{(j)} \neq 0},
	\end{align}
for all pairs of parties $(i, j)$ with $i \neq j$.
If in addition one of the remaining indices $\vec{l}$ and $\vec{m}$ equals $\vec{0}$, the conditions are fulfilled in a trivial way. Hence, we only need to consider the case where no index in Eq. (\theequation) equals $\vec{0}$.

The necessary condition given in Eq.  (\theequation) can be rewritten as
	\begin{align}		
		\label{equ:4qubitconditions}
		\vec{h}^{(i)} \left( \vec{h}^{(j)} \right)^T \odot (N_1 - N_2) = 0,
	\end{align}
where $\odot$ denotes the Hadamard product, $N_1 \equiv \vec{\eta}\vec{\eta}^T$ and $\left[N_2\right]_{\vec{l} \vec{m}} \equiv \eta_{\vec{l}+\vec{m}}$. Moreover, using this notation, the necessary condition given in Eq. (\ref{equ:gi}) is equivalent to 
\begin{align}
\label{equ:validpovm}
	 \eta_{\vec{u}} h_{\vec{u}}^{(i)} = g_{\vec{u}}^{(i)} \ \forall \vec{u}.
\end{align}

Note that Eq. (\ref{equ:4qubitconditions}) and Eqs. (\ref{equ:validpovm}) basically coincide with the necessary conditions for the existence of a SEP transformation for the generic four qubit states (see \cite{article:mes}), but here $\eta_{\vec{k}}$ are complex whereas in the four qubit case $\eta_{\vec{k}}$ are real numbers. However, we will see that in the three qutrit scenario there are additional solutions to this equation leading to the fact that some states are reachable via SEP but not LOCC  as we will see in section \ref{sec:sepneqlocc}.

\section{Reachable states via SEP}
We will now use the conditions given in Eq. (\ref{equ:sepconfinal1}) and (\ref{equ:sepconfinal2}) and the standard form to characterize states which are reachable via SEP, i.e., states which are final states of a (non-trivial) SEP transformation, and thus final states for which Eq. (\ref{equ:sepconfinal1}) can be fulfilled such that the initial state is not LU-equivalent to the final state.

The necessary and sufficient condition for SEP-convertibility given in Eq. (\ref{equ:sepconfinal1}) and the necessary condition for SEP-convertibility given in Eq. (\ref{equ:sepconfinal2}) are written in such a way that they only depend on the final state $h \ket{\psi}$. Note that all these conditions account for equations in the $\eta_{\vec{k}}$ which have to hold depending on the form of $h$. If we want to know whether a state $h \ket{\psi}$ can be reached via SEP we have to know whether there exits a set of probabilities $\{p_{\vec{k}}\}_{\vec{k}}$ fulfilling the conditions, i.e., solving all the equations in the $\eta_{\vec{k}}$ which have to hold for a given $h$. We will call such a set $\{p_{\vec{k}}\}_{\vec{k}}$ a solution. Each solution $\{p_{\vec{k}}\}_{\vec{k}}$ which fulfills all the conditions accounts for a possible SEP-transformation transforming some state into $h \ket{\psi}$. Moreover, for each solution  $\{p_{\vec{k}}\}_{\vec{k}}$ the initial state $g \ket{\psi}$ can be easily determined via Eq. (\ref{equ:gi}) and therefore depends on $h$ and the probabilities $\{p_{\vec{k}}\}_{\vec{k}}$. Let us now state some simple observations which we will then use to derive the reachable states.

\begin{observation}
\label{obs:etakisone}
If one of the $\eta_{\vec{k}}$ with $\vec{k} \neq \vec{0}$ fulfills that $|\eta_{\vec{k}}| = 1$ , then three probabilities $p_{\vec{i}}$ sum up to one and the other six are zero. Which of the probabilities have to be zero depends on  $\vec{k}$.
\end{observation}
\begin{proof}
Recall that $\eta_{\vec{k}} \in \conv{1, e^{i\frac{2 \pi}{3}}, e^{i\frac{4 \pi}{3}}} \ \forall {\vec{k}}$. For $\vec{k}\in \mathbb{Z}_3^2\backslash\{\vec{0}\}$ the weights in the convex combination are certain sums of three probabilities. Therefore $|\eta_{\vec{k}}| = 1$ implies that one weight in the convex combination has to be one. Hence the corresponding three probabilities have to sum up to one. This proves the observation.
\end{proof}

In the next observation we consider the case where more than one of the independent $\eta_{\vec{k}}$ has absolute value equal to one.
\begin{observation}
\label{obs:moreetakisone}
If more than one of the independent $\eta_{\vec{k}}$ fulfills that $|\eta_{\vec{k}}| = 1$, then one $p_{\vec{i}} = 1$ (and the others vanish).
\end{observation}
\begin{proof} To prove this observation, note that due to Observation \ref{obs:etakisone} two triples of probabilities have to sum up to one, and the other probabilities have to vanish as two independent $\eta_{\vec{k}}$ have absolute value equal to one. One readily sees that for any allowed choice of the triples of probabilities which sum up to one there is at most one probability which appears in both of these triplets. Therefore one probability $p_{\vec{k}}$ has to be one, and all the others have to vanish.
\end{proof}
One can easily see from Eq. (\ref{equ:gour}) that if more than one of the independent $\eta_{\vec{k}}$ have absolute value equal to one, the standard forms of $g \ket{\psi}$ and $h \ket{\psi}$ coincide and hence, the intial state is LU-equivalent to the final state.  Hence, in order to study non-trivial transformations we only need to consider those solutions where at most one $\eta_{\vec{k}}$ (and its complex conjugate $\eta_{\vec{k}}^* = \eta_{-\vec{k}}$) have absolute value equal to 1.

In the next observation we consider the consequences of Observation \ref{obs:etakisone} and Observation \ref{obs:moreetakisone} on the vectors $\vec{h}^{(i)}$.
\begin{observation}
\label{obs:atmostone}
 If there exists $\vec{k} \in \mathbb{Z}_3^2\backslash\{\vec{0}\}$ and $i, j \in \{1,2,3\}$, $i\neq j$ such that $ h_{\vec{k}}^{(i)} h_{\vec{k}}^{(j)} \neq 0$, then the only solutions to Eq. (\ref{equ:sepconfinal2}) are those, where $ \left|  \eta_{\vec{k}} \right| = 1$.
\end{observation}
\begin{proof}
Whenever there are two parties $(i,j)$ for which the same components in the vectors $\vec{h}^{(i)}$ and $\vec{h}^{(j)}$ are non vanishing, i.e., $h_{\vec{k}}^{(i)} \neq 0$ and $h_{\vec{k}}^{(j)} \neq 0$ (which also implies that $h_{-\vec{k}}^{(i)} \neq 0$ and $h_{-\vec{k}}^{(j)} \neq 0$) for the same $\vec{k}$, then  due to  Eq. (\ref{equ:sepconfinal2}) both $\eta_{\vec{k}} \eta_{\vec{k}} = \eta_{-\vec{k}}$ and $\eta_{\vec{k}} \eta_{-\vec{k}} = \eta_{\vec{0}}$ must hold. Due to Eq. (\ref{equ:eta0}) the latter condition is equivalent to  $|\eta_{\vec{k}}|^2 = 1$.
\end{proof}
Note that Observation \ref{obs:atmostone} and Observation \ref{obs:moreetakisone} imply that for any non-trivial solution of Eq. (\ref{equ:gour}) the vectors $\vec{h}^{(i)}$ have to be of a form that
 $h_{\vec{k}}^{(i)} h_{\vec{k}}^{(j)} \neq 0$  holds for at most one $\vec{k}$ (and the corresponding $-\vec{k}$), i.e., there exists no $i_1 \neq j_1$, $ i_2 \neq j_2$, $\vec{0} \neq \vec{k}_1 \neq  \pm \vec{k}_2 \neq \vec{0}$ such that
\begin{align}
	h_{\vec{k}_1}^{(i_1)} h_{\vec{k}_1}^{(j_1)} \neq 0 \ \text{and} \  h_{\vec{k}_2}^{(i_2)} h_{\vec{k}_2}^{(j_2)} \neq 0.
\end{align}

We will now characterize the states which are reachable via SEP. In particular we will make use of Eq. (\ref{equ:sepconfinal1}), the necessary and sufficient, and Eq. (\ref{equ:sepconfinal2}), the necessary condition for SEP-reachability.
Recall that in order to transform a state $g\ket{\psi}$ into the state $h \ket{\psi}$ via SEP, the parties have to implement the POVM  $\left\{ M_{\vec{k}} \right\}_{\vec{k} \in \mathbb{Z}_3^2}$, where $M_{\vec{k}} = \frac{1}{\sqrt{r}} \sqrt{p_{\vec{k}}} h S_{\vec{k}}^{\otimes 3} g^{-1}$ and that we use the notation $h_{\vec{w}}^i \in \Span{\identity, S_{\vec{w}}, S_{-\vec{w}}}$.
Using the observations above we are now able to prove the following theorem.
\begin{theorem}
\label{theo:sepreach}
A generic state $h \ket{\psi}$ is reachable via SEP from some other LU-inequivalent state iff (up to permutations) either
\begin{enumerate}[(i)]
	\item $h = h_1 \otimes h_2 \otimes \identity$ such that $\vec{h}^{(1)} \odot \vec{h}^{(2)} = 0$, where not both   $h_1 \propto \identity$ and   $h_2 \propto \identity$ or
	\item $h = h_1 \otimes h_{\vec{w}}^2 \otimes h_{\vec{w}}^3$, for some $\vec{w} \in \mathbb{Z}_3^2 \backslash \{\vec{0}\}$ with $h_1 \neq h^1_{\vec{w}}$.
\end{enumerate}
\end{theorem}
\begin{proof}
\emph{Only if:} We first show that states which can be reached non--trivially are necessarily of the form given in the theorem.
To this end, we show that the conditions given in Eq. (\ref{equ:sepconfinal1}) can be fulfilled for non LU--equivalent initial and final states, $g \ket{\psi}$ and $h \ket{\psi}$, respectively, only if the final states are of the form given above.

Due to Observation \ref{obs:moreetakisone} we know that for non--trivial transformations not more than one of the independent $\eta_{\vec{k}}$ can have absolute value equal to 1. We distinguish the case where no (exactly one) $\eta_{\vec{k}}$ has absolute value equal to 1 respectively. Note that here and in the following $\vec{k}\neq \vec{0}$. We will see that these two cases correspond to the cases (i) and (ii) of the theorem, respectively.

Let us first consider the case where no $\eta_{\vec{k}}$ has absolute value equal to 1. First we consider the case that for no party we have $\vec{h}^{(i)} = \vec{0}$, i.e., $\forall i \ \exists {\vec{k}} \quad  h_{\vec{k}}^{(i)} \neq 0$. Then just by the fact that there cannot exist two parties such that components $h_{\vec{k}}^{(i)} \neq 0$ for the same $\vec{k}$ (as otherwise our assumption $|\eta_{\vec{k}}| \neq 1$ would not hold due to Observation \ref{obs:atmostone}) we can conclude that wlog for party 2 (3) there can be at most one index $\vec{k}_2$ such that $\ h_{\pm \vec{k}_2}^{(2)} \neq 0$ ($\vec{k}_3$ such that $\ h_{\pm \vec{k}_3}^{(3)} \neq 0$) respectively, where $\vec{k}_2 \neq \pm \vec{k}_3$ and for party 1 only the remaining two independent components can be non vanishing, i.e., up to joint permutations of the  entries in the vectors we have
$\vec{h}^{(1)} = \left(  h_{\left( \begin{smallmatrix}  1 \\  0  \end{smallmatrix}\right)}^{(1)}, h_{\left( \begin{smallmatrix}  2 \\  0  \end{smallmatrix}\right)}^{(1)}, h_{\left( \begin{smallmatrix}  0 \\  1  \end{smallmatrix}\right)}^{(1)}, h_{\left( \begin{smallmatrix}  0 \\  2  \end{smallmatrix}\right)}^{(1)}, 0,0,0,0  \right)^T$,
$\vec{h}^{(2)} = \left(  0,0,0,0, h_{\left( \begin{smallmatrix}  1 \\  1  \end{smallmatrix}\right)}^{(2)}, h_{\left( \begin{smallmatrix}  2 \\ 2  \end{smallmatrix}\right)}^{(2)}, 0,0  \right)^T$ and
$\vec{h}^{(3)} = \left(  0,0,0,0,0,0, h_{\left( \begin{smallmatrix}  2 \\ 1  \end{smallmatrix}\right)}^{(3)} , h_{\left( \begin{smallmatrix}  1 \\  2  \end{smallmatrix}\right)}^{(3)}   \right)^T$.
It is straightforward to show that due to Eq. (\ref{equ:sepconfinal1}) it must hold that
\begin{align}
	\label{equ:etatilde}
	\eta_{\tilde{\vec{k}}_1} \eta_{\tilde{\vec{k}}_2} \eta_{\tilde{\vec{k}}_3} &= \eta_{\vec{0}} = 1,
\end{align}
where $\tilde{\vec{k}}_i$ is chosen out of the set $\left\{ \pm \vec{k}_i \right\}$ such that $\tilde{\vec{k}}_1+\tilde{\vec{k}}_2+\tilde{\vec{k}}_3=\vec{0}$ holds. Note that such a triple of vectors always exists, even if there is only one independent non vanishing component in the vector $\vec{h}^{(1)}$.
Using that $\left| \eta_{\vec{k}} \right|\leq 1$ for all $\vec{k}$ in Eq. (\theequation) we obtain that it must hold that $|\eta_{\tilde{\vec{k}}_1}| = |\eta_{\tilde{\vec{k}}_2}| = |\eta_{\tilde{\vec{k}}_3}| = 1$ which leads to a contradiction to the assumption that no $\eta_{\vec{k}}$ has absolute value equal to 1. Therefore, wlog it must hold that $\vec{h}^{(3)} = \vec{0}$. As can be easily seen Eq. (\ref{equ:sepconfinal1}) (or equivalently Eq. (\ref{equ:sepconfinal2})) can then only be fulfilled for $\vec{h}^{(1)}, \vec{h}^{(2)}$ such that $\vec{h}^{(1)} \odot \vec{h}^{(2)} = 0$ [as in Theorem \ref{theo:sepreach} (i)] as none of the $\eta_{\vec{k}}$ has absolute value equal to 1. Note that if $\vec{h}^{(1)}= \vec{h}^{(2)}=0$ the initial state $g \ket{\psi}$ would be LU-equivalent to the final state $h\ket{\psi}$, which is the seed state. In all the other instances the states given in Theorem \ref{theo:sepreach} (i) can be reached from LU--inequivalent states choosing $ p_{\vec{k}} = \frac{1}{9} \  \forall {\vec{k}}$, which leads to $  \eta_{\vec{k}} = 0 \ \forall {\vec{k} \neq \vec{0}}$ \footnote{Note that there might exist additional solutions $\left\{p_{\vec{k}} \right\}_{\vec{k}}$. However, for some choices of $h$ the solution  $\left\{p_{\vec{k}} \right\}_{\vec{k}}$ with $ p_{\vec{k}} = \frac{1}{9} \ \forall {\vec{k}}$ is indeed the only solution, i.e., we can reach the state  $h \ket{\psi}$ only from the seed state (see Section \ref{sec:sepneqlocc}).}.

Let us now consider the case where exactly one of the independent $\eta_{\vec{k}}$ has absolute value equal to 1.
We will first consider the case that there exist two parties  (wlog party 1 and party 2) for which at least two components in the corresponding vectors $\vec{h}^{(i)}$ are non vanishing. We denote the indices of these components by $\vec{k}_1$, $\vec{l}_1$ and $\vec{k}_2$, $\vec{l}_2$, respectively. As we will show in the following, this always leads to the existence of at least two indices $\vec{k}'$ and $\vec{l}'$ with $\vec{k}' \neq \pm \vec{l}'$ s.t. $\left| \eta_{\vec{k}'} \right| = \left| \eta_{\vec{l}'} \right| = 1$ contradicting the assumption. As $h_{\vec{k}_1}^{(1)} h_{-\vec{k}_1}^{(1)} h_{\vec{l}_1}^{(1)} h_{-\vec{l}_1}^{(1)} \neq 0$ and $h_{\vec{k}_2}^{(2)} h_{-\vec{k}_2}^{(2)} h_{\vec{l}_2}^{(2)} h_{-\vec{l}_2}^{(2)} \neq 0$, where $ \vec{k}_1 \neq \pm \vec{l}_1$ and  $\vec{k}_2 \neq \pm \vec{l}_2$ Eq. (\ref{equ:sepconfinal2}) leads to quadratic equations in $\eta_{\vec{i}}$ which have to be fulfilled.
We now distinguish the following cases:
\begin{enumerate}[(a)]
	\item Either both of the indices $\pm \vec{k}_1$ and $\pm \vec{l}_1$ coincide with the indices $\pm \vec{k}_2$ and $\pm \vec{l}_2$,
	\item or exactly one of the indices $\pm \vec{k}_1$ and $\pm \vec{l}_1$ coincides with one of the indices $\pm \vec{k}_2$ and $\pm \vec{l}_2$,
	\item or none of the indices $\pm \vec{k}_1$ and $\pm \vec{l}_1$ coincide with $\pm \vec{k}_2$ or $\pm \vec{l}_2$.
\end{enumerate}
In case (a), the components $h_{\vec{k}_1}^{(1)}$, $h_{\vec{l}_1}^{(1)}$ and $h_{\vec{k}_1}^{(2)}$, $h_{\vec{l}_1}^{(2)}$ are non--vanishing, which implies (see Observation \ref{obs:atmostone}) that $\left| \eta_{\vec{k}_1} \right| = \left| \eta_{\vec{l}_1} \right| = 1$, which contradicts the assumption.

In case (b), we can assume wlog that $\vec{k}_1 \in\{ \pm \vec{k}_2\}$, however, $\vec{l}_1 \not\in \{\pm \vec{l}_2\}$. Due to  Observation \ref{obs:atmostone} it must then hold that  $\left| \eta_{\vec{k}_1} \right| = 1$. From Eq. (\ref{equ:sepconfinal2}) we get the following set of quadratic equations in $\eta_{\vec{k}}$
\begin{align}
	\eta_{\pm \vec{k}_1} \eta_{\pm \vec{l}_1} &= \eta_{\pm \vec{k}_1 \pm \vec{l}_1}\nonumber \\
	\eta_{\pm \vec{k}_1} \eta_{\pm \vec{l}_2} &= \eta_{\pm \vec{k}_1 \pm \vec{l}_2} \nonumber\\
	\eta_{\pm \vec{l}_1} \eta_{\pm \vec{l}_2} &= \eta_{\pm \vec{l}_1 \pm \vec{l}_2}.
\end{align}
In the following, we will show that $\eta_{\vec{k}_1}$ occurs on the right hand side of at least one of these equations, implying that there exists some $\vec{l}\neq \vec{k}_1$ such that $\left| \eta_{\vec{l}} \right| = 1$, contradicting the assumption. As $\vec{l}_1, \vec{k}_1 \neq \vec{0}$ we have $\left( \vec{k}_1 + \vec{l}_1 \right) \neq \pm \left( \vec{k}_1 - \vec{l}_1 \right)$. The same is true for  $\vec{k}_1 + \vec{l}_2$ and $\vec{k}_1 - \vec{l}_2$. Moreover $\left\{\pm \vec{k}_1 \pm \vec{l}_1 \right\} = \left\{\pm \vec{k}_1 \pm \vec{l}_2 \right\}$ is not possible as otherwise $\vec{l}_1 = \pm \vec{l}_2$ leading to a contradiction. Therefore, three of the four independent $\eta_{\vec{i}}$ have to appear on the right hand side of the first two equations in Eq. (\theequation). Now we have that either   $\left\{ \pm \vec{k}_1 \pm \vec{l}_1, \pm \vec{k}_1 \pm \vec{l}_2  \right\} = \mathbb{Z}_3^2$, which implies that $\eta_{\vec{k}_1}$ occurs on the right hand side of at least one of the Eqs. (\theequation), or one of the indices $\pm \vec{k}_1 \pm \vec{l}_1$ coincides with one of the indices $\pm \vec{k}_1 \pm \vec{l}_2$. In the latter case out of the many possible combinations of indices coinciding, there are only two fundamentally different cases, either  $\vec{k}_1 + \vec{l}_1 = \vec{k}_1 + \vec{l}_2$ or $\vec{k}_1 + \vec{l}_1 = - \vec{k}_1 + \vec{l}_2$ \footnote{This is due to the fact that wlog we can redefine $\vec{k_1} \rightarrow -\vec{k_1}$, $\vec{l_1} \rightarrow -\vec{l_1}$, and/or $\vec{l_2} \rightarrow -\vec{l_2}$ such that all the other cases are transformed into one of these two.}. The first case is not possible as  $ \vec{l}_1 \neq \pm \vec{l}_2$ due to our assumption. In the second case we obtain $\vec{l}_1 - \vec{l}_2 = \vec{k}_1$, implying that we get $\vec{k}_1$ as an index on the right hand side of the third equation in Eq. (\theequation). Hence, case (b) cannot occur.

In case (c), we have $\left\{ \pm \vec{k}_1, \pm \vec{l}_1\right\} \cap \left\{ \pm \vec{k}_2 ,\pm \vec{l}_2 \right\} = \emptyset$. Again, due to Eq. (\ref{equ:sepconfinal2}) we get a set of quadratic equations in the $\eta_{\vec{l}}$. As before, we find that all independent $\eta_{\vec{l}}$ appear in the linear term of these equations. Using similar methods as before, as well as Eq. (\ref{equ:etaabs}), and the fact that in the here considered case  $|\eta_{\vec{k}}|=1$ one obtains that all $\eta_{\vec{l}}$ must have absolute value equal to 1, leading again to a contradiction.

As none of the cases (a)--(c) can occur, we have that only for one party (wlog party 1) there can be more than one independent non vanishing component $h_{\pm\vec{k}}^{(i)} \neq 0$.  For the remaining, however, it must hold that $h_2 = h_{\vec{w}_2}^2$ and $h_3 = h_{\vec{w}_3}^3$ for some $\vec{w}_2, \vec{w}_3 \in \mathbb{Z}_3^2\setminus \{\vec{0}\}$ (see Observation \ref{observation:span}). Moreover, as already mentioned if $\vec{h}^{(i)}=0$ $\forall i$ the corresponding state is not reachable via SEP. Thus, there has to be at least one party $i$ for which $\vec{h}^{(i)}$ has a non-vanishing component and we choose wlog $i=1$.  We show now that either $\vec{w}_2 = \vec{w}_3 = \vec{w} $ has to hold (up to particle permutations) or that one can always choose $\vec{w}_2 = \vec{w}_3 = \vec{w}$. Note that if $\vec{h}^{(2)}=0$ we can trivially choose $\vec{w}_2 = \vec{w}_3$  and analogous for $\vec{h}^{(3)}=0$. Let us now consider the case that neither $\vec{h}^{(2)}=0$ nor $\vec{h}^{(3)}=0$.  Due to Eq. (\ref{equ:sepconfinal1}) there cannot exist $\vec{k}_1, \vec{k}_2,\vec{k}_3$ such that $\vec{k}_1 \neq \pm \vec{k}_2 \neq \pm \vec{k}_3 \neq \pm \vec{k}_1$ and $h_{\vec{k}_1}^{(1)} h_{\vec{k}_2}^{(2)} h_{\vec{k}_3}^{(3)} \neq 0$, as otherwise more than one independent $\eta_{\vec{k}}$ would have absolute value equal to 1. This can be seen using the same arguments as before, which imply that the different vectors can always be added to obtain $\eta_{\vec{0}}=1$ on the right hand side of Eq. (\ref{equ:etatilde}). Hence, for any triple of vectors, $\vec{k}_1 , \vec{k}_2 ,\vec{k}_3 $ for which $h_{\vec{k}_1}^{(1)} h_{\vec{k}_2}^{(2)} h_{\vec{k}_3}^{(3)} \neq 0$ holds, it must be that $\vec{k}_i\in\{\pm\vec{k}_j\}$ for some $\{i,j\}\in\{\{1,2\},\{1,3\},\{2,3\}\}$. This implies that in case party $1$ has more than one independent non vanishing component  $h_2 = h_{\vec{w}}^2$  and $h_3 = h_{\vec{w}}^3$ has to hold for some $\vec{w} \in \mathbb{Z}_3^2\setminus \{\vec{0}\}$. In case party $1$ has exactly one independent non-vanishing component we choose wlog that $\vec{h}^{(2)}$ is parallel to $\vec{h}^{(3)}$. Note that in case $h_1 = h_{\vec{w}}^1$ and not both $ h_{\vec{w}}^2$ and $ h_{\vec{w}}^3$ are proportional to the identity using Observation \ref{obs:atmostone} one can easily see that the standard forms of $H$ and $G$ would coincide implying that the initial state is LU-equivalent to the final state.
The remaining states are of the form corresponding to case (ii) of the theorem \footnote{Note that some states also belong to case (i) of the theorem.}.
Hence, we have that a state $h \ket{\psi}$ is reachable via SEP from some other LU-inequivalent state only if it can be written in one of the forms given in the theorem.

\emph{If:} We now show, that states of the form given in the theorem can be reached from some LU-inequivalent state via SEP.
The SEP maps which can be used to reach the states can be easily constructed as follows. In general the SEP maps are given in Eq. (\ref{equ:sepmap}) (recall that $r=1$).
Here, $G_i$ can be obtained by depolarizing $H_i$ with the generalized Pauli matrices $S_{\vec{k}}$ weighted with the probabilities $p_{\vec{k}}$ which are the solution of Eq.  (\ref{equ:gour}).
States of the form given in case (i) of Theorem \ref{theo:sepreach} can be reached from the (non LU--equivalent) seed state $\ket{\psi}$. The corresponding SEP map is given by $\left\{ M_{\vec{k}} \right\}_{\vec{k}\in \mathbb{Z}_3^2}$, where
\begin{align}
	\label{equ:caseisep}
	M_{\vec{k}} = \frac{1}{3} \left( h_1 \otimes h_2 \otimes \identity \right) S_{\vec{k}}^{\otimes 3}.
\end{align}

States of the form given in case (ii) of Theorem \ref{theo:sepreach} can be reached from the (non LU--equivalent) state $g_{\vec{w}}^1 \otimes h_{\vec{w}}^2 \otimes h_{\vec{w}}^3 \ket{\psi}$, where $G_1 = \left(g_{\vec{w}}^1 \right)^\dagger g_{\vec{w}}^1 = \frac{1}{3} \sum_{\vec{k} \in \{\vec{0}, \vec{w}, -\vec{w}\}} S_{\vec{k}}^\dagger H_1 S_{\vec{k}}$. The corresponding SEP map is given by $\left\{ M_{\vec{k}} \right\}_{\vec{k} \in \{\vec{0}, \vec{w}, -\vec{w}\}}$, with
\begin{align}
	\label{equ:caseiisep}
	M_{\vec{k}} = \frac{1}{\sqrt{3}} \left(h_1 \otimes \identity \otimes \identity \right) S_{\vec{k}}^{\otimes 3}  \left(\left(g_{\vec{w}}^1\right)^{-1} \otimes\identity \otimes \identity \right).
\end{align}

Here, we used that $S_{\pm \vec{w}}$ commutes with $h^i_{\vec{w}}$. Hence, all states given in Theorem \ref{theo:sepreach} can be reached non--trivially, which completes the proof.
\end{proof}

\section{Example of a pure state transformation that is possible via SEP but not via LOCC}
\label{sec:sepneqlocc}

In the previous section we have characterized all states reachable via SEP. As SEP reachability is only a necessary condition for LOCC reachability not all of those states have to be reachable via LOCC. In this section we will characterize a family of states which is, surprisingly, not reachable via LOCC, although the family is of the form given in case (i) of Theorem \ref{theo:sepreach} and therefore reachable via SEP.

In order to characterize this family of states, we first prove the following lemma.
\begin{lemma}
\label{lemma:noloccreach}
  Let
	$h = h_1 \otimes h_2 \otimes \identity$ be such that (up to permutations)
	 $H_i \in \Span{\identity, S_{\vec{u}_i}, S_{-\vec{u}_i}, S_{\vec{w}_i}, S_{-\vec{w}_i}}$ for $i \in \{1,2\}$,
where $\{ \pm \vec{u}_1, \pm \vec{w}_1, \pm \vec{u}_2, \pm \vec{w}_2, \vec{0} \} = \mathbb{Z}_3^2$ and none of the coordinates of $H_i$ corresponding to $S_{\vec{u}_i}$ and $S_{\vec{w}_i}$  is zero.
Then, the only possibility to reach a generic state $h \ket{\psi}$ via SEP are SEP transformations where the initial state is the seed state $\ket{\psi}$.
\end{lemma}
\begin{proof}
We have shown in Theorem \ref{theo:sepreach} that states of the form given in this lemma are reachable via SEP as Eq. (\ref{equ:gour}) can be solved by choosing $\eta_{\vec{k}} = 0 \ \forall {\vec{k} \neq 0}$ which corresponds to a probability distribution $p_{\vec{k}} = \frac{1}{9} \ \forall {\vec{k}}$. We will show, that this is the only solution to Eq. (\ref{equ:gour}). This leads to the fact that the state $h \ket{\psi}$ can only be the final state of a SEP protocol if the initial state is the seed state $\ket{\psi}$ as the map $\Lambda(\rho) = \sum_{\vec{k}} p_{\vec{k}} S_{\vec{k}}^\dagger \rho S_{\vec{k}}$ with $p_{\vec{k}} = \frac{1}{9} \ \forall {\vec{k}}$ is completely depolarizing.

Due to the form of $h$ specified in the lemma and due to the necessary conditions given in Eq. (\ref{equ:sepconfinal2}), the following quadratic equations in the $\eta_{\vec{k}}$ have to hold:
\begin{align}
	\label{equ:etalist}
	\eta_{\pm \vec{u}_1} \eta_{\pm \vec{u}_2} &= \eta_{\pm \vec{u}_1 \pm \vec{u}_2} \nonumber \\
	\eta_{\pm \vec{u}_1} \eta_{\pm \vec{w}_2} &= \eta_{\pm \vec{u}_1 \pm \vec{w}_2} \nonumber\\
	\eta_{\pm \vec{w}_1} \eta_{\pm \vec{u}_2} &= \eta_{\pm \vec{w}_1 \pm \vec{u}_2} \nonumber\\
	\eta_{\pm \vec{w}_1} \eta_{\pm \vec{w}_2} &= \eta_{\pm \vec{w}_1 \pm \vec{w}_2}.
\end{align}
As $\vec{u}_1 +\vec{u}_2\not\in\left\{ \pm \vec{u}_1, \pm \vec{u}_2\right\}$, we have that $\vec{u}_1 +\vec{u}_2\in\left\{ \pm \vec{w}_1, \pm \vec{w}_2\right\}$. Consider wlog that $\vec{u}_1 +\vec{u}_2=\vec{w}_2$. Thus, we have that $\eta_{\vec{u}_1} \eta_{ \vec{u}_2} = \eta_{ \vec{w}_2}$ and $\eta_{-\vec{u}_1} \eta_{ \vec{w}_2} = \eta_{ \vec{u}_2}$ [see Eqs. (\ref{equ:etalist})]. Inserting the second equation into the first one one obtains that either $|\eta_{\vec{u}_1}|=1$ or $\eta_{ \vec{w}_2}=0$ and therefore $\eta_{ \vec{u}_2}=0$.  Note that as $\left\{ \pm \vec{u}_1, \pm \vec{w}_1, \pm \vec{u}_2, \pm \vec{w}_2, \vec{0} \right\} = \mathbb{Z}_3^2$ we have that all the independent $\eta_{\vec{k}}$ appear at least once on the right hand side of Eqs. (\theequation). Thus, it can be easily seen that for $|\eta_{\vec{u}_1}|=1$ it has to hold that  $|\eta_{\vec{k}}|=1$ for at least two other independent $\vec{k}$.
However, $\left| \eta_{\vec{k}} \right| = 1$ for more than one $\vec{k}$ would imply that the initial state is LU equivalent to the final state due to Observation \ref{obs:moreetakisone}. Thus, we only have to consider the solution $\eta_{ \vec{w}_2}=\eta_{ \vec{u}_2}=0$. It can be easily shown using again the fact that  all the independent $\eta_{\vec{k}}$ appear at least once on the right hand side of Eqs. (\theequation) that  this implies $\eta_{\vec{k}} = 0 \ \forall {\vec{k}}$. This solution corresponds to a probability distribution  $\left\{ p_{\vec{k}}\right\}_{\vec{k} \in \mathbb{Z}_3^2}$ with $p_{\vec{k}} = \frac{1}{9} \ \forall {\vec{k}}$.
Therefore $\left\{ p_{\vec{k}}\right\}_{\vec{k} \in \mathbb{Z}_3^2}$ with $p_{\vec{k}} = \frac{1}{9} \ \forall {\vec{k}}$ is indeed the only way in which Eq. (\ref{equ:gour}) can be solved.
\end{proof}

Note that for states of the form given in Lemma \ref{lemma:noloccreach} it holds that there exists no $S_{\vec{k}}\in S(\ket{\Psi})$ with $\vec{k}\in\mathbb{Z}_3^2\backslash \{\vec{0}\}$ s.t $S_{\vec{k}}^\dagger H_iS_{\vec{k}}=H_i$ for $i\in\{1,2\}$.

Using the lemma above we can now prove the following theorem.
\begin{theorem}
\label{theo:noloccreach}
States of the form given in Lemma \ref{lemma:noloccreach} are not reachable via a LOCC protocol from some other LU-inequivalent state (although they can be reached via SEP).
\end{theorem}
\begin{proof}
As proven in Lemma \ref{lemma:noloccreach}, the state $h \ket{\psi}$ can only be the final state of a SEP protocol if the initial state was the seed state $\ket{\psi}$. Therefore also every deterministic LOCC protocol transforming a state into the final state $h \ket{\psi}$ has as initial state the state $\ket{\psi}$.
Using this result, we will first show that we cannot transform $\ket{\psi}$ into $h \ket{\psi}$ deterministically in a single round protocol, i.e., only one party does something nontrivial, communicates the measurement outcomes to the other parties and the other parties are allowed to apply LUs. Then we will show that we cannot reach $h \ket{\psi}$ with any (even an infinitely many round) LOCC protocol.

Let us assume now that we can reach $h \ket{\psi} =   h_1 \otimes h_2 \otimes \identity \ket{\psi} $ via LOCC in one round deterministically \footnote{As in this case the symmetries of the seed states are unitary, it follows trivially that the transformation can not be done in a single round. Note, however, that in general this does not need to hold \cite{article:spee15}.}.
Let us assume that party 1 does something nontrivial. Then there has to exist an operator $A$ and unitaries $U_2, U_3$ such that $A \otimes U_2 \otimes U_3 \ket{\psi} \propto h_1 \otimes h_2 \otimes \identity \ket{\psi}$. However, using the fact that the symmetries are unitary this implies that $H_2 \propto \identity$, which is a  contradiction to the assumption. Obviously, similar arguments can be used to find a contradiction if party 2 or party 3 were the parties applying nontrivial operations.  Hence, we cannot reach  $h \ket{\psi}$ in one round deterministically.

We now show that there cannot be a finitely many round or even infinitely many round LOCC protocol transforming $\ket{\psi}$ into the final state $h \ket{\psi}$.
Assume that we can transform the initial state $\ket{\psi}$ into $h\ket{\psi}$. As we cannot reach the final state in one single round, as shown above, it must hold that there exists a round $m$ and a branch $i$ such that there exists a state $\ket{\psi_{i,m}}$ different from the initial state $\ket{\psi}$ and the final state $h \ket{\psi}$, i.e.,
\begin{align}
	\exists\ket{\psi_{i,m}}: \ \ket{\psi_{i,m}} \not\simeq_{LU} \ket{\psi},\ \ket{\psi_{i,m}} \not\simeq_{LU} h \ket{\psi},
\end{align}
which can be transformed to the final state $h\ket{\psi}$, as after any round $m$ and for any branch $i$ the remaining (finitely or infinitely many round) part of the protocol has to transform the state $\ket{\psi_{i,m}}$ into $h \ket{\psi}$ deterministically. Using Lemma \ref{lemma:noloccreach}, this, however, implies that either $\ket{\psi_{i,m}} \simeq_{LU} \ket{\psi}$ or $\ket{\psi_{i,m}} \simeq_{LU} h \ket{ \psi}$ which is a contradiction to Eq. (\theequation).
Hence, the state $h \ket{\psi}$ cannot be reached via (finitely or infinitely many round) LOCC.
\end{proof}

As these transformations can be implemented via SEP but not via LOCC, this implies the existence of a familiy of entanglement monotones that can increase under SEP. In fact the maximum success probability for reaching a state of the form given in Lemma \ref{lemma:noloccreach} via LOCC, which has been shown to be an entanglement monotone \cite{article:monotones}, can increase under SEP.

\section{The MES of generic three qutrit states}
In this section, we will characterize the states which are reachable via LOCC. Note that these are all the states reachable via SEP, characterized in Theorem \ref{theo:sepreach}, excluding those for which we have proven that they are not reachable via LOCC in Theorem \ref{theo:noloccreach}.

\begin{theorem}
\label{theo:loccreach}
A generic state $h \ket{\psi}$ is reachable via LOCC from some other LU-inequivalent state iff (up to permutations)
	$h = h_1 \otimes h_{\vec{w}}^2 \otimes h_{\vec{w}}^3$ for some $\vec{w} \in \mathbb{Z}_3^2$, where $h_1 \neq h^1_{\vec{w}}$.
\end{theorem}
\begin{proof}
\emph{Only if:}
Combining the results of Theorem \ref{theo:sepreach} and Theorem \ref{theo:noloccreach} we get that states reachable via LOCC necessarily have to be of the form given in this theorem.

\emph{If:}
We will prove this direction by explicitly giving the initial state $g \ket{\psi}$ and the LOCC protocol transforming $g \ket{\psi}$ into the final state $h \ket{\psi}$. We will distinguish in the following the two different cases: (i) not both $h_{\vec{w}}^2$ and $h_{\vec{w}}^3$ are proportional to the identity and (ii) $h_{\vec{w}}^2=h_{\vec{w}}^3\propto\one$.
For case (i) consider the following LOCC protocol.
The parties start out with the initial state $g_{\vec{w}}^1 \otimes h_{\vec{w}}^2 \otimes h_{\vec{w}}^3 \ket{\psi}$, where
$\left(g_{\vec{w}}^1 \right)^\dagger g_{\vec{w}}^1 = \frac{1}{3} \sum_{\vec{k} \in \{\vec{0}, \vec{w}, -\vec{w}\}} S_{\vec{k}}^\dagger H_1 S_{\vec{k}}$.
Party $1$ measures $\left\{ M_{\vec{k}}\right\}_{\vec{k} \in \{\vec{0}, \vec{w}, -\vec{w}\}}$, where $M_{\vec{k}} = \frac{1}{\sqrt{3}} h_1 S_{\vec{k}} \left(g_{\vec{w}}^1\right)^{-1}$. She communicates the outcome $\vec{k}$ to party $2$ and $3$ and they apply, depending on $\vec{k}$, the unitary $S_{\vec{k}}$.  It can be easily seen that $\left\{ M_{\vec{k}}\right\}_{\vec{k} \in \{\vec{0}, \vec{w}, -\vec{w}\}}$ is a valid POVM. Moreover, using the fact that $h_{\vec{w}}^i$ commutes with $S_{\pm\vec{w}}$ one easily sees that in every branch of the protocol the parties end up with the final state $h \ket{\psi}$
\footnote{Note that for a proper choice of $\vec{w}\in \mathbb{Z}_3^2 \setminus \{\vec{0}\}$ this LOCC protocol also allows to reach states of the form given by case (ii). Note further that the corresponding initial states are not necessarily in the MES.}.
For case (ii) the states can be reached from the corresponding seed state via the following LOCC protocol. Party $1$ implements a POVM whose elements are given by $M_{\vec{k}} = \frac{1}{3} h_1 S_{\vec{k}}$ where $\vec{k}\in \mathbb{Z}_3^2$. Then, in case of the measurement outcome $\vec{k}$ party $2$ and $3$ apply the unitary $S_{\vec{k}}$.  This proves the theorem.
\end{proof}

Any generic state which cannot be written as those given in Theorem \ref{theo:loccreach} is necessarily in the MES. However, all states which are reachable are reachable from a state in the MES. This can be easily seen as follows. Except for states that are (up to particle permutations) of the form  $ h_1 \otimes h_{\vec{w}}^2 \otimes \identity \ket{\Psi}$  with $ \vec{w} \in \mathbb{Z}_3^2 \backslash \{\vec{0}\}$, $\vec{h}^{(1)} \odot \vec{h}^{(2)} = 0$ and $\vec{h}^{(2)}\neq 0$  the corresponding LOCC protocols (and initial states) are given in the proof of Theorem \ref{theo:loccreach}. States of the form $ h_1 \otimes h_{\vec{w}}^2 \otimes \identity \ket{\Psi}$ can be reached from the seed state by implementing the following LOCC protocol. Party $2$ implements the POVM $\{M_{\vec{k}}\}_{\vec{k}\in \mathbb{Z}_3^2}$ where $M_{\vec{k}}= \frac{1}{3} h_{\vec{w}}^2 S_{\vec{k}}$ and depending on the measurement outcome the other parties apply $S_{\vec{k}}$ (see case (ii) in the proof of Theorem \ref{theo:loccreach}). The resulting state is given by  $ \one \otimes h_{\vec{w}}^2 \otimes \identity \ket{\Psi}$. Then, the protocol given for case (i) in the proof of Theorem \ref{theo:loccreach} is implemented in order to reach the state $ h_1 \otimes h_{\vec{w}}^2 \otimes \identity \ket{\Psi}$.   Thus, any state that is reachable can be obtained via LOCC from a state that can not be written as in Theorem \ref{theo:loccreach}. Therefore, the states which are not reachable via LOCC constitute the MES for generic three qutrit states as stated in the following corollary.
Note that in particular the seed state $\ket{\psi}$ is contained in $MES_{(3,3,3)}$.
\begin{corollary}
\label{corollary:mes}
The MES for generic three qutrit states is given by
 $MES_{(3,3,3)} = \left\{ h_1 \otimes h_2 \otimes h_3 \ket{\psi} \right\}$, where $h_1 \otimes h_2 \otimes h_3$ cannot be written as (up to permutations)  $h_1 \otimes h_{\vec{w}}^2 \otimes h_{\vec{w}}^3$ for some $\vec{w} \in \mathbb{Z}_3^2$, where $h_1 \neq h^1_{\vec{w}}$.
\end{corollary}

The states in Theorem \ref{theo:loccreach} are characterized by 16 real parameters implying that the set of reachable states is of measure zero (as it was the case in the four qubit scenario). Thus, the MES is of full measure. However, we will now show that most of the states in the MES are isolated, i.e., they can neither be reached from, nor converted into any other LU-inequivalent, truly tripartite, $\vec{d}$-level entangled state. In order to do so, we will characterize convertible states in the following theorem.

\begin{theorem}
\label{theo:loccconv}
A generic state $g \ket{\psi}$ is convertible via LOCC to some other LU-inequivalent state iff (up to permutations)
	$g = g_1 \otimes g_{\vec{w}}^2 \otimes g_{\vec{w}}^3$, for some $\vec{w} \in \mathbb{Z}_3^2$, where $g_1$ is arbitrary.
\end{theorem}

\begin{proof}
\emph{If:} We show that $g\ket{\psi}$, where $g$ is of the form given in the theorem, i.e., $g = g_1 \otimes g_{\vec{w}}^2 \otimes g_{\vec{w}}^3$, can always be converted into a LU-inequivalent state $h \ket{\psi} =  h_1 \otimes g_{\vec{w}}^2 \otimes g_{\vec{w}}^3 \ket{\psi}$ for some $h_1$. To show this we construct the LOCC protocol accomplishing this transformation. Consider  $\{  M_{\vec{0}}, M_{\vec{w}}, M_{-\vec{w}}  \}$, which we will later show is a valid POVM, where
\begin{align}
	M_{\vec{k}} = \sqrt{p_{\vec{k}}} h_1 S_{\vec{k}} \left( g_1\right)^{-1} \otimes S_{\vec{k}}^{\otimes 2}
\end{align}
and where $\{ p_{\vec{0}}, p_{\vec{w}}, p_{-\vec{w}}\}$ is a proper probability distribution, i.e., $p_{\vec{k}} \geq 0$ and $\sum_{\vec{k}} p_{\vec{k}} = 1$. It is LOCC as it can be implemented by party 1 measuring and party 2 and 3 applying (depending on the measurement outcome) local unitaries $S_{\vec{0}}$ or $S_{\pm \vec{w}}$. As $g_{\vec{w}}^{2}$ ($g_{\vec{w}}^{3}$) commutes with $S_{\pm \vec{w}}$ the protocol indeed transforms $g\ket{\psi}$ into $h\ket{\psi}$. We now confirm that for any $g$ we can find $h$ and $\{ p_{\vec{0}}, p_{\vec{w}}, p_{-\vec{w}}\}$ such that $\{  M_{\vec{0}}, M_{\vec{w}}, M_{-\vec{w}}  \}$ is a valid POVM, i.e. $\sum_{\vec{k}} M_{\vec{k}}^\dagger M_{\vec{k}} = \identity$. This is equivalent to
\begin{align}
	p_{\vec{0}} H_1 + p_{\vec{w}} S_{\vec{w}}^\dagger H_1 S_{\vec{w}} +  p_{-\vec{w}} S_{-\vec{w}}^\dagger H_1 S_{-\vec{w}} = G_1,
\end{align}
which is fulfilled iff [see Eq. (\ref{equ:validpovm})]
\begin{align}
\label{equ:validpovm1}
	 \eta_{\vec{u}} h_{\vec{u}}^{(1)} = g_{\vec{u}}^{(1)} \ \forall \vec{u}.
\end{align}
Note that here $\eta_{\vec{0}} = \eta_{\vec{w}} = \eta_{-\vec{w}} = 1$ as there are only three non vanishing probabilities $p_{\vec{0}}$, $p_{\vec{w}}$, and $p_{-\vec{w}}$.
Let us first consider the case where $g_1 = g_{\vec{w}}^1$. As we have seen in the proof of Theorem $\ref{theo:loccreach}$ we can transform this state to states of the form $h_1 \otimes g_{\vec{w}}^2 \otimes g_{\vec{w}}^3 \ket{\psi}$ by choosing $p_{\vec{0}} = p_{\vec{w}} = p_{-\vec{w}} = \frac{1}{3}$.
Let us now consider the case where $g_1\neq g_{\vec{w}}^1$.
 Obviously, we can find an operator $H$ satisfying the condition given in Eq. (\ref{equ:validpovm1}). However, we have to make sure that $H$ corresponds to a valid state, i.e., $H$ is positive.
To see that this is always possible consider a probability distribution $\{p_{\vec{0}}, p_{\vec{w}}, p_{-\vec{w}}\}$ with $p_{\vec{0}} = 1 - \frac{2\epsilon}{3}$, $p_{\vec{w}} = \frac{\epsilon}{3}$, $p_{-\vec{w}} = \frac{\epsilon}{3}$, where $\epsilon > 0$. With this probability distribution we obtain $\eta_{\vec{u}} = 1-\epsilon \ \forall \vec{u}$ with $\vec{u} \neq \vec{0} ,\pm \vec{w}$.
This choice of the probabilities increases the components of the initial state $g_{\vec{u}}^{(1)}\ \forall \vec{u}$ with $  \vec{u} \neq \vec{0},\pm \vec{w}$ by a factor of $\frac{1}{1-\epsilon}$.
Note that the expression for the eigenvalues of the operator $H$ is continuous in $\epsilon$. Using this as well as the fact that the operator $G$ is strictly positive, we obtain that given any state $g \ket{\psi}$ of the form above, there is always a probability distribution $\{p_{\vec{0}}, p_{\vec{w}}, p_{-\vec{w}}\}$ which leads to a small enough $\epsilon$ such that $H$ still corresponds to a valid state.
Hence, states of the form given in the theorem are convertible via LOCC.

\emph{Only if:} We now show that convertible states are necessarily of the form given in the theorem.
Due to Theorem \ref{theo:loccreach} we know that reachable states have to be of the form $h\ket{\psi}$ with $h = h_1 \otimes h_{\vec{w}}^2 \otimes h_{\vec{w}}^3$  for $ h_{\vec{w}}^i \in \Span{\identity, S_{\vec{w}}, S_{-\vec{w}}} $ and $h_1 \neq h^1_{\vec{w}}$. Thus any LOCC-convertible state $g \ket{\psi}$ can only be converted into states of this form.
By inserting $H_i$ ($G_i$) expressed in the generalized Pauli matrices $\{S_{\vec{k}}\}_{\vec{k}}$ as in Eq. (\ref{equ:pauli}) into
the necessary condition for SEP inter-convertibility given in Eq. (\ref{equ:gi}) we obtain
	\begin{align}
		 \eta_{\vec{k}} h_{\vec{k}}^{(i)} = g_{\vec{k}}^{(i)} \ \forall {\vec{k}}.
	\end{align}
Hence, a component of $G_i$ can only be non vanishing if the corresponding component of $H_i$ is non vanishing. Thus, $g$ has to be of the form given in the theorem, i.e., $g = g_1 \otimes g_{\vec{w}}^2 \otimes g_{\vec{w}}^3$, which completes the proof.
\end{proof}

Note that the set of convertible states via LOCC coincides with the set of convertible states via SEP. This is due to the fact that for any state that is reachable via SEP (see Theorem \ref{theo:sepreach}) it has been shown that they can only be reached via SEP from the seed states or from a state of the form given in Theorem \ref{theo:loccconv} (see proof of that theorem).
As mentioned above, the MES is of full measure.
However, with  Theorem \ref{theo:loccconv} we see that only a subset of the states in the MES are states which are convertible. They constitute a family defined by 10 real parameters of the form $g_{\vec{w}}^1 \otimes g_{\vec{w}}^2 \otimes g_{\vec{w}}^3 \ket{\psi}$ for some $\vec{w} \in \mathbb{Z}_3^2$ where no $g_{\vec{w}}^i\propto \one$ (except for the seed state). Note that all other states in the MES are isolated, i.e., they can neither be reached from, nor converted to any other LU-inequivalent state via LOCC. Although the MES is of full measure, the subset of non-isolated states, which are the physically more relevant ones, is of measure zero (as it was the case in the four qubit setting).

\section{Conclusion and Outlook}

In this paper we characterized the MES for generic three qutrit states. It turns out that as in the generic four qubit case deterministic LOCC transformations are hardly ever possible and almost all states are isolated. We identified the zero measure subset of convertible states in the MES. These states are the most relevant ones for deterministic entanglement manipulation.\\
In order to characterize the MES we studied first reachability under deterministic SEP transformations. Interestingly, we found a family of pure states that can only be reached via SEP from the corresponding seed state. We showed that this implies that these transformations can not be implemented via LOCC (even if one allows for infinitely many rounds of communication). Up to our knowledge these are the first examples of transformations among pure states that can be accomplished via SEP but not via LOCC. The existence of these transformations implies that the maximum success probability of reaching the final state of the transformation via LOCC, which is an entanglement monotone \cite{article:monotones}, increases under SEP \footnote{For similar quantities see \cite{article:chitambar09}}.\\
As the possibility for deterministic LOCC transformation is very rare it would be interesting to study the multi-copy case as well as the use of catalysts in order to see whether isolation still remains a generic feature.  Another interesting question would be to consider $\epsilon$-LOCC transformations, i.e., allowing for a small deviation in the final state. Our study of deterministic LOCC transformations and the tools used in this paper also open the way to determine the entanglement measures introduced in \cite{article:schwaiger15} for generic three qutrit states.\\ 

\section*{ACKNOWLEDGEMENTS}
The research was funded by the Austrian Science Fund (FWF): Y535-N16.

\appendix*
\section{Symmetries of the three qutrit seed states}	
\label{app:sym}

In this appendix, we will calculate the symmetries of the generic three qutrit seed state $\ket{\psi}$ given in Eq. (\ref{equ:seed}). Note that some parameter choices in the seed state are excluded as discussed in Section \ref{sec:seed}. In order to improve readability we recall here Lemma  \ref{lemma:symmetries}.
\\

\noindent\textit{ {\bf Lemma \ref{lemma:symmetries}.}
Let $\ket{\psi}$ be a generic three qutrit seed state. The symmetries of $\ket{\psi}$ are given by the so-called generalized Pauli matrices,
		\begin{align}		
				S(\ket{\psi})= &\left\{\identity_3^{\otimes 3}, (X)^{\otimes 3}, (X^2)^{\otimes 3}, (Z)^{\otimes 3}, (Z^2)^{\otimes 3}, \right.\nonumber\\
						 &\quad \left. (XZ)^{\otimes 3}, (XZ^2)^{\otimes 3}, (X^2Z)^{\otimes 3}, (X^2Z^2)^{\otimes 3}   \right\},
		\end{align}
		where
		\begin{align}
				X = \begin{pmatrix}
					 0 & 1 & 0 \\
					 0 & 0 & 1 \\
					 1 & 0 & 0
					 \end{pmatrix} 	
				\quad\text{and}\quad
				Z = \begin{pmatrix}
					 1 & 0 & 0 \\
					 0 & \exp{i\frac{2\pi}{3}} & 0 \\
					 0 & 0 & \exp{i\frac{4\pi}{3}}
					 \end{pmatrix}. 	
		\end{align}}

\begin{proof}
To prove this Lemma we will assume that $S$ is a symmetry of a generic seed state $\ket{\psi}$. We will write $S$ in terms of local operators as $S = A \otimes B \otimes C$ with $A, B, C \in  \operatorname{GL}(3, \mathbb{C})$. The condition for $S$ to be a symmetry is given by
\begin{align}
	\label{eq:sym}
	A \otimes B \otimes C \ket{\psi} = \ket{\psi}.
\end{align}
In the following, we will denote the matrix elements of  $A$, $B$, and $C$ in the computational basis by $A_{i j}$, $B_{i j}$, or $C_{i j}$ respectively.
We will now derive necessary conditions for $A \otimes B \otimes C $ to be a symmetry. Later on we will use these necessary conditions to prove that (up to normalization of $A$, $B$, and $C$) for generic seed states only $S_{\vec{k}}^{\otimes 3}$,  where $S_{\vec{k}}$ denotes the generalized Pauli matrices, can be symmetries.
Let us introduce the following bipartite states for party 2 and 3:
\begin{align}
\ket{\phi_1} &= c^*\ket{1 2} - b^* \ket{2 1} \nonumber\\
\ket{\phi_2} &= c^*\ket{2 0} - b^* \ket{0 2} \nonumber\\
\ket{\phi_3} &= c^*\ket{0 1} - b^* \ket{1 0} \nonumber\\
\ket{\phi_4} &= a^*\ket{1 2} - b^* \ket{0 0} \nonumber\\
\ket{\phi_5} &= a^*\ket{2 0} - b^* \ket{1 1} \nonumber\\
\ket{\phi_6} &= a^*\ket{0 1} - b^* \ket{2 2} \nonumber\\
\ket{\phi_7} &= a^*\ket{2 1} - c^* \ket{0 0} \nonumber\\
\ket{\phi_8} &= a^*\ket{0 2} - c^* \ket{1 1} \nonumber\\
\ket{\phi_9} &= a^*\ket{1 0} - c^* \ket{2 2}.
\end{align}
Note that they are all orthogonal to the seed state $\ket{\psi}$, i.e., for all $i$ we have $\phantom{.}_{23}\braket{\phi_i}{\psi}_{123}=0$.
 Projecting Eq. (\ref{eq:sym}) onto the states $\ket{\phi_i}_{23}$ yields
 \begin{align}
 	A\otimes\phantom{.}_{23}\braopket{\phi_i}{B \otimes C}{\psi}_{123} \propto \phantom{.}_{23}\braket{\phi_i}{\psi}_{123} = 0 \ \forall i,
 \end{align}
 which is equivalent to
 \begin{align}
 	\label{eq:projection}
 	\phantom{.}_{23}\braopket{\phi_i}{B \otimes C}{\psi}_{123} = 0 \ \forall i,
 \end{align}
as $A$ is an invertible operator. Note that Eq. (\theequation) is a set of 9 vector equations which are linear in entries of each of the matrices $B$ and $C$ and quadratic in the seed parameters $a$, $b$, and $c$. As the equations are linear in the entries of $C$, we can write these equations as
 \begin{align}
 	X  \begin{pmatrix}
					 \vec{c}_0  \\
					 \vec{c}_1  \\
					 \vec{c}_2
					 \end{pmatrix},
 \end{align}
where $\vec{c}_i^T$ is the $i$th row of the matrix $C$. Note that $X$ is of a special form as it can be written as
 \begin{align}
 	X =  \begin{pmatrix}
					 0 & -b M_2 & c M_1 \\
					 c M_2 & 0 & -b M_0 \\
					 -b M_1 & c M_0 & 0 \\
					 -b M_0 & 0 & a M_1 \\
					 a M_2 & -b M_1 & 0 \\
					 0 & a M_0 & -b M_2 \\
					 -c M_0 & a M_2 & 0 \\
					 0 & -c M_1 & a M_0 \\
					 a M_1 & 0 & -c M_2
					 \end{pmatrix} 	
 \end{align}
 with
  \begin{align}
 	M_i = \begin{pmatrix}
					 a & c & b \\
					 b & a & c \\
					 c & b & a
					 \end{pmatrix} \odot 	
		\begin{pmatrix}
					 B_{i 0} & B_{i 2} & B_{i 1} \\
					 B_{i 2} & B_{i 1} & B_{i 0} \\
					 B_{i 1} & B_{i 0} & B_{i 2}
					 \end{pmatrix}.
 \end{align}
 This leads to the following set of equations (excluding the non-generic choices for the seed parameters where $a=0$, $b=0$, or $c=0$)
 \begin{flalign}
 	\label{eq:symeqs}
 	M_0 \vec{c}_1 &= \frac{b}{c} M_1 \vec{c}_0 & M_1 \vec{c}_1 &= \frac{a}{b} M_2 \vec{c}_0 & M_2 \vec{c}_1 &= \frac{c}{a} M_0 \vec{c}_0 \nonumber \\
 	M_0 \vec{c}_2 &= \frac{c}{a} M_1 \vec{c}_1 & M_1 \vec{c}_2 &= \frac{b}{c} M_2 \vec{c}_1 & M_2 \vec{c}_2 &= \frac{a}{b} M_0 \vec{c}_1 \nonumber \\
 	M_0 \vec{c}_0 &= \frac{a}{b} M_1 \vec{c}_2 & M_1 \vec{c}_0 &= \frac{c}{a} M_2 \vec{c}_2 & M_2 \vec{c}_0 &= \frac{b}{c} M_0 \vec{c}_2.
 \end{flalign}
 Note that the three vectors $\vec{c}_0$, $\vec{c}_1$, and $\vec{c}_2$ form a basis of $\mathbb{C}^3$ as the matrix $C$ is invertible. Together with the set of equations given in Eq. (\theequation) we have that either all matrices $M_i$ are invertible or none of them is. To see this note that $\operatorname{range}(M_i) = \Span{M_i \vec{b}_0, M_i \vec{b}_1,M_i \vec{b}_2}$, where $\left\{ \vec{b}_0, \vec{b}_1, \vec{b}_2\right\}$ is an arbitrary basis. Hence, using the fact that  $\left\{\vec{c}_0, \vec{c}_1, \vec{c}_2 \right\}$ forms a basis and using the equations given in Eq. (\theequation) we have $\operatorname{range}(M_0) = \operatorname{range}(M_1) = \operatorname{range}(M_2)$. In particular, we have that $\operatorname{rk}(M_0) = \operatorname{rk}(M_1) = \operatorname{rk}(M_2)$, where $\operatorname{rk}$ denotes the rank. Hence, either all of them are invertible (have full rank) or none of them is.
 In the following, we will distinguish these two cases.
 \begin{enumerate}[(a)]
 \item None of the matrices $M_i$ is invertible.
 \item All of the matrices $M_i$ are invertible.
 \end{enumerate}

 Let us first consider case (a). We will make use of the adjugates of the matrices $M_i$ which are denoted by $\operatorname{adj}(M_i)$. Note that the adjugate of a matrix is always defined and can be constructed by calculating the minors of the matrix independently of whether a matrix is invertible or not. In case of the matrices $M_i$ we have
 \begin{widetext}
 \begin{align}
 \operatorname{adj}(M_i) =
		\begin{pmatrix}
					 a^2 B_{i 1} B_{i 2} - b c B_{i 0}^2 & b^2 B_{i 0} B_{i 1} - a c B_{i 2}^2 & c^2 B_{i 0} B_{i 2} - a b B_{i 1}^2 \\
					 c^2 B_{i 0} B_{i 1} - a b B_{i 2}^2 & a^2 B_{i 0} B_{i 2} - b c B_{i 1}^2 & b^2 B_{i 1} B_{i 2} - a c B_{i 0}^2 \\
					 b^2 B_{i 0} B_{i 2} - a c B_{i 1}^2 & c^2 B_{i 1} B_{i 2} - a b B_{i 0}^2 & a^2 B_{i 0} B_{i 1} - b c B_{i 2}^2
					 \end{pmatrix}.
 \end{align}
 \end{widetext}
Using that $\operatorname{adj}(M_i) M_i = \operatorname{det}(M_i) \identity$ together with Eq. (\ref{eq:symeqs}) and using the fact that $\left\{\vec{c}_0, \vec{c}_1, \vec{c}_2 \right\}$ forms a basis it has to hold that
 \begin{align}
 	\operatorname{adj}(M_i) M_j &= 0,
 \end{align}
 where $i, j \in \left\{0,1,2 \right\}$ and $i \neq j$.
 Excluding the cases where $a^3 + b^3 + c^3 = 0$ or $\left(a^3 + b^3 + c^3 \right)^3 = (3 a b c)^3$ we find that the equations we get by requiring that for all $i$ we have $\operatorname{det}(M_i) = 0$ together with the equations we get by taking the traces over Eq. (\theequation) contradicts the assumption that $B$ is invertible. It is thus, furthermore, a contradiction to the assumption that $S$ is a symmetry.

 Let us now consider case (b), where all matrices $M_i$ are invertible. We can thus multiply the first (second, third) three equations in Eq. (\ref{eq:symeqs}) with $M_0^{-1}$ ($M_1^{-1}$, $M_2^{-1}$) from the left respectively. As the vectors $\vec{c}_i$ form a basis, we can conclude that there exists a basis transformation $T$ such that
 \begin{align}
 T M_0^{-1} M_1 T^{-1} =  \begin{pmatrix}
					 0 & 0 & \frac{b}{a} \\
					 \frac{c}{b} & 0 & 0 \\
					 0 & \frac{a}{c} & 0
					 \end{pmatrix}, 	
 \end{align}
and the matrices  $T M_1^{-1} M_2 T^{-1}$ and $T M_2^{-1} M_0 T^{-1}$ are matrices where the non vanishing entries of the matrix given in Eq. (\theequation)  are cyclically permuted. Note that the eigenvalues of these matrices are $1$, $e^{i\frac{2\pi}{3}}$, and  $e^{-i\frac{2\pi}{3}}$. This implies that the eigenvalues of the matrices $M_0^{-1} M_1$, $M_1^{-1} M_2$, and $M_2^{-1} M_0$ are, independently of the normalization of $B$ and independently of the seed parameters, also $1$, $ e^{i\frac{2\pi}{3}}$, and  $e^{-i\frac{2\pi}{3}}$. Using the special form of the matrices $M_i$ in the computational basis and imposing the known form of the characteristic polynomial of $M_i^{-1}M_j$ (as we know the eigenvalues), one obtains a set of equations for the entries of $B$ independently of the entries of $C$. The conditions imposed on the matrix elements of $B$ by these equations are necessary for the matrix $B$ to be part of a symmetry. Excluding the parameter choices $a^3 + b^3 + c^3 = 0$ and $\left(a^3 + b^3 + c^3 \right)^3 = (3 a b c)^3$ the equations lead to one of the following two cases. We find that either, in case (b1), six of the entries in $B$ have to vanish while the remaining ones are (normalizing the matrix to $\operatorname{det}(B)=1$) taken out of the set $\left\{1, \omega, \omega^2 \right\}$, where $\omega =e^{i\frac{2\pi}{3}}$ or, in case (b2), the matrix $B$ is (normalizing the matrix to $B_{00} = 1$) one of the following 162 candidates
\begin{align}
 \begin{pmatrix}
					 1 & \omega^{i} & \omega^{j} \\
					 \omega^{k} & \omega^{k+i+m} & \omega^{k+j+2m} \\
					 \omega^{l} & \omega^{l+i+2m} & \omega^{l+j+m}
					 \end{pmatrix},
\end{align}
where $i,j,k,l \in \left\{0,1,2\right\}$ and $m \in \left\{1,2\right\}$.

Let us first consider the case where $B$ is of the form given in case (b1).
Excluding the parameter choices $a^9=b^9$, $a^9=c^9$, and $b^9=c^9$, using Eq. (\ref{eq:projection}), and assuming that $B$ is of the form given in case (b1) we have that $C \propto B$ and that only the generalized Pauli matrices remain as possible symmetries. It remains to derive the form of $A$. It is straightforward to verify that whenever  $A \otimes B \otimes C \ket{\psi} \propto \ket{\psi}$, we also have $B \otimes C \otimes A \ket{\psi} \propto \ket{\psi}$ due to the fact that $\ket{\psi}$ exhibits a cyclically particle permutation symmetry. Using this we find that in case (b1) $S = A \otimes B \otimes C \in \left\{S_i^{\otimes 3}\right\}_i$, where $S_i$ are the generalized Pauli matrices. Furthermore, it is straightforward to verify that they actually are symmetries of the seed state.

Let us now consider the case where $B$ is one of the candidates given in case (b2) [see Eq. (\theequation)]. Let us first derive necessary conditions for the entries of the matrix $C$. Obviously, the same arguments as above apply to derive necessary conditions for the matrix $C$. Thus, also $C$ has to be of one of the forms we derived for $B$, i.e., either six of the entries in $C$ vanish or $C$ is one of the candidates given in Eq. (\theequation). The former cannot be the case as this, as we have seen before, leads to the fact that $B \propto C$ contradicting the assumption that $B$ is of the form given in case (b2). Hence, also $C$ has to be one of the candidates given in Eq. (\theequation).
Using Eq. (\ref{eq:projection}) it is now straightforward to verify that none of the candidate combinations for matrices $B$ and $C$ turns out to be a symmetry as long as the parameter choices $a^3=b^3$, $a^3=c^3$, $b^3=c^3$, $a+b+c=0$, $a+ \omega b+c=0$, $a+ \omega^2 b+c=0$, $a+b+ \omega c=0$, $a+b+ \omega^2 c=0$, $a+\omega b+ \omega^2 c=0$, $a+\omega^2 b+ \omega c=0$, $a b + b c  + c a = 0$, $a b + \omega b c  + c a = 0$, $a b + \omega^2 b c  + c a = 0$, $a b + b c  + \omega c a = 0$, $a b + b c  + \omega^2 c a = 0$, $a b + \omega b c  + \omega^2 c a = 0$, and $a b + \omega^2 b c  + \omega c a = 0$ are excluded. Hence, we get no additional symmetries in case (b2). This proves the lemma.
\end{proof}

\bibliography{references}

\end{document}